\documentclass[sigconf]{acmart}
\AtBeginDocument{%
  }
\usepackage{enumitem}
\usepackage{pifont}

\usepackage{multirow}
\usepackage{colortbl}

\definecolor{teacherblue}{HTML}{EAF5FC}
\definecolor{regerblue}{HTML}{C6E6F6}

\newcommand{\mstd}[2]{%
    \mbox{#1{\footnotesize\ensuremath{\pm}#2}}%
}

\makeatletter
\let\originalfnsymbol\@fnsymbol
\renewcommand{\@fnsymbol}[1]{%
  \ifnum#1=1\relax
    \originalfnsymbol{2}%
  \else
    \originalfnsymbol{#1}%
  \fi
}
\makeatother

\begin{document}

% \title{From Continuous to Discrete: Modeling inter-item Relations with Semantic IDs for Generative Recommendation}
\title{REGER: Modeling inter-item Relations with Semantic IDs for Generative Recommendation}

\author{Yuxi Huang}
\affiliation{
  \institution{University of Electronic Science and Technology of China}
  \city{Chengdu}
  \country{China}
}
\email{yuxi.h.cs@gmail.com}

\author{Silin Zhou}
\affiliation{
  \institution{University of Electronic Science and Technology of China}
  \city{Chengdu}
  \country{China}
}
\email{zhousilinxy@gmail.com}

\author{Chenhao Wang}
\affiliation{
  \institution{University of Electronic Science and Technology of China}
  \city{Chengdu}
  \country{China}
}
\email{chenhao.wang@std.uestc.edu.cn}

\author{Zihan Li}
\affiliation{
  \institution{University of Electronic Science and Technology of China}
  \city{Chengdu}
  \country{China}
}
\email{202621081431@std.uestc.edu.cn}

\author{Gangming Guo}
\affiliation{
  \institution{University of Electronic Science and Technology of China}
  \city{Chengdu}
  \country{China}
}
\email{guogm1016@gmail.com}

\author{Pei Ke}
\affiliation{
  \institution{University of Electronic Science and Technology of China}
  \city{Chengdu}
  \country{China}
}
\email{kepei1106@outlook.com}

\author{Lisi Chen}
\affiliation{
  \institution{University of Electronic Science and Technology of China}
  \city{Chengdu}
  \country{China}
}
\email{lchen012@e.ntu.edu.sg}

\author{Shen Gao}
\authornote{Shen Gao and Shuo Shang are corresponding authors.}
\affiliation{
  \institution{University of Electronic Science and Technology of China}
  \city{Chengdu}
  \country{China}
}
\email{shengao@uestc.edu.cn}

\author{Shuo Shang}
\authornotemark[1]
\affiliation{
  \institution{University of Electronic Science and Technology of China}
  \city{Chengdu}
  \country{China}
}
\email{jedi.shang@gmail.com}

\begin{abstract}
Personalized recommendation depends on learning semantic and collaborative relations among items from online content and user interactions. Generative recommendation (GR) offers a promising paradigm by autoregressively generating discrete semantic IDs (SIDs) of items. However, randomly initialized SID token embeddings and GR models trained from scratch do not directly inherit the semantic relations inherent in items, while inter-item collaborative relations are learned largely implicitly through token-level attention and next-token supervision. To address this limitation, we propose \textbf{REGER}, a unified framework that explicitly incorporates inter-item relations into semantic pretraining, history encoding, and SID generation. We pretrain the GR model through generative semantic retrieval, enabling it to inherit inter-item semantic similarity by learning to generate semantic neighbors' SIDs. We combine global attention with local attention aligned with SID boundaries to capture historical context and directed transitions between adjacent items. 
% We further develop a theoretical framework for hierarchical preference decomposition over the SID prefix tree, proving an exact factorization of collaborative preference distributions and a weighted decomposition of item-level KL divergence into prefix-conditional KL divergences. Building on this analysis, we distill a discriminative teacher's full conditional distributions over all legal next tokens at ground-truth prefixes along the target SID path, jointly optimizing this objective with cross-entropy. 
We formalize item-level preference distillation on the SID prefix tree, deriving an exact factorization of item preference distributions into prefix conditionals and a weighted decomposition of item-level KL divergence into conditional KL terms.
Guided by this decomposition, we renormalize the generator's probabilities over legal next tokens and match the resulting distributions to the teacher's conditional preferences along the ground-truth SID path, jointly optimizing the distillation objective with cross-entropy.
Extensive experiments on real-world datasets validate the effectiveness of the proposed preference distillation and demonstrate that REGER consistently outperforms existing GR baselines while retaining efficient inference. 
% % The code is available at \url{https://anonymous.4open.science/r/REGER/}.
% \end{abstract}
The code will be made publicly available upon acceptance.
\end{abstract}


%%
% Keywords. The author(s) should pick words that accurately describe
% the work being presented. Separate the keywords with commas.


\maketitle

\section{Introduction}
\label{Introduction}
\begin{figure}[t]
    \centering
    \includegraphics[width=1\linewidth]{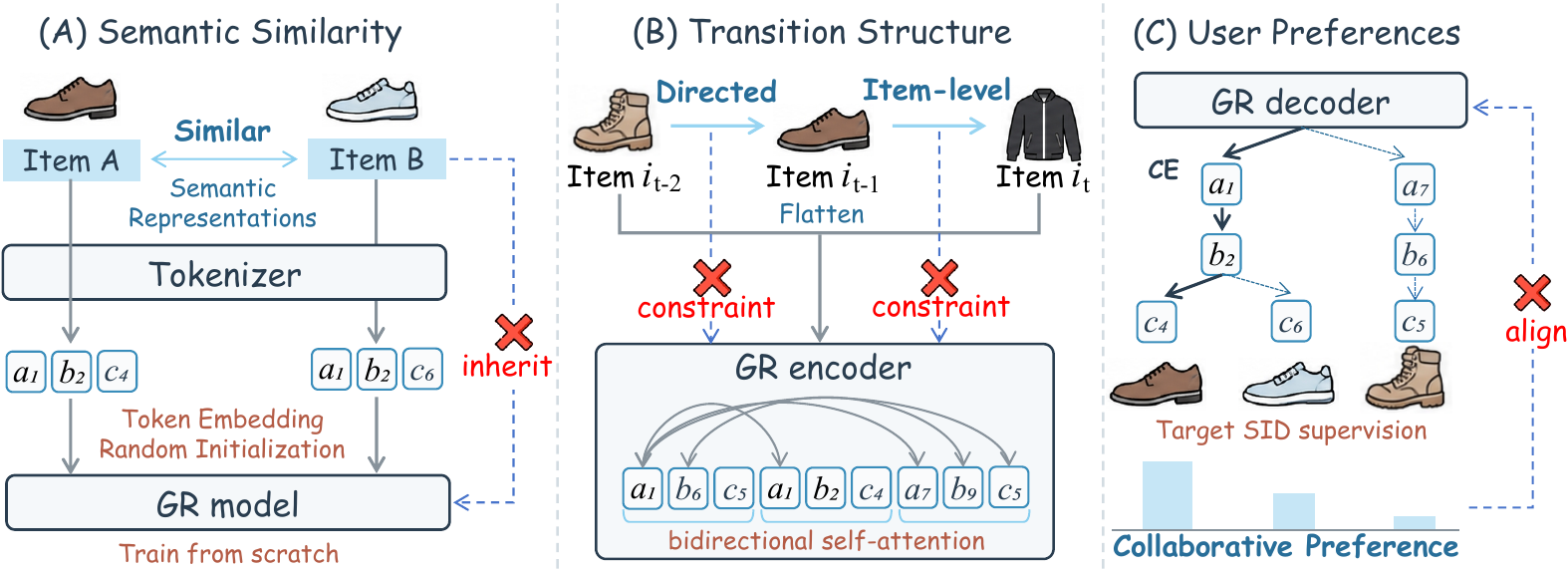}
    \caption{Challenges in explicitly modeling inter-item semantic and collaborative relations in a standard GR framework:
(A) semantic similarity relations arising from item text,
(B) directed transition relations from the previous to the current item in the interaction history, and
(C) ranking relations among candidate items induced by user preferences.}
    \label{fig:intro}
\end{figure}
% Recommender systems play a key role in mitigating information overload on the Web by learning user preferences from historical interactions and helping users discover relevant products and content across e-commerce, video, and social media platforms~\cite{Xie2025LOHRec,Lin2026CARE}.
Generative recommendation (GR)~\cite{Ju2025Handbook,Hou2026GRSurvey} has emerged as a promising paradigm that captures user preferences end-to-end by autoregressively generating discrete token sequences representing items. Unlike traditional discriminative methods that learn continuous collaborative embeddings for individual items, GR typically uses residual quantization to map pretrained item representations into compact semantic identifiers (SIDs), where common prefixes reflect semantic similarity and shared tokens facilitate knowledge transfer across items~\cite{Rajput2024TIGER,Xia2026ACERec}. Realizing this potential requires effective modeling of semantic and collaborative relations among items in discrete token sequences.

Existing efforts have largely focused on improving item tokenization to better capture item relations. Some methods draw on subword tokenization to merge fine-grained tokens into larger units based on co-occurrence patterns~\cite{Hou2025Actionpiece,Miao2026SST}, while others incorporate collaborative information into SIDs by aligning semantic representations with collaborative item embeddings~\cite{Xiao2026UNGER,Zhao2026MusicRec}. Recent work also combines GR with dense matching from traditional recommenders through hybrid retrieval, exploiting the complementary strengths of the two paradigms~\cite{Yang2025LIGER,Ding2026MemGen-GR}. Furthermore, GFlowGR~\cite{Wang2026GFlowGR} incorporates collaborative signals into item-level rewards to guide SID generation through flow balance objectives, while TopoTok~\cite{Liu2026ToPoTok} preserves item neighborhood relations during tokenization through topology distillation at multiple granularities.
% Despite these advances, an inherent bottleneck remains in GR: \textbf{how to incorporate inter-item relations in continuous embedding space into the modeling and generation of discrete SID token sequences.} 
Despite these advances, an inherent bottleneck remains in GR: \textbf{how to incorporate inter-item relations captured in continuous embedding spaces into the modeling and generation of discrete SID token sequences.}
% Fig.~\ref{fig:intro} illustrates how this bottleneck manifests in a standard encoder-decoder GR framework.

Modeling inter-item relations is central to capturing user preferences in sequential recommendation~\cite{Wang2023Causer,Kim2026PRISM}. Item text provides semantic cues for identifying potential substitutes and complements~\cite{McAuley2015Substitutable}, while co-occurrence and transition patterns in user interaction histories provide collaborative signals that characterize behavioral associations between items~\cite{Wang2020GCE-GNN}.
Moreover, differences in a user's preference strengths across candidate items give rise to relative preference relations among them, providing the basis for personalized ranking~\cite{Lee2021DCD}.
These relations are defined at the item level, whereas GR uses SID tokens as its basic modeling units, leaving inter-item relations to be learned implicitly.

% As illustrated in Fig~\ref{fig:intro}(A), although shared prefixes preserve some semantic structure, randomly initialized token embeddings and a GR model trained from scratch do not directly inherit the semantic similarity relations captured by continuous item representations. During history encoding, user interaction histories are flattened into SID token sequences, and the encoder integrates token information through bidirectional self-attention. However, its attention structure does not explicitly impose local constraints that respect complete SID boundaries and capture directed transitions between adjacent items, as illustrated in Fig~\ref{fig:intro}(B). During generation, GR relies on next-token cross-entropy (CE) to maximize the likelihood of the target SID. However, this single-target supervision does not fully exploit the set of relevant candidate items or explicitly align generation with collaborative preference relations between users and candidate items, leaving inter-item competitive relations without explicit constraints, as shown in Fig.~\ref{fig:intro}(C). GR therefore needs to explicitly model inter-item relations when learning item semantics, encoding interaction histories, and generating recommendations.
As shown in Fig.~\ref{fig:intro}, this bottleneck manifests as three challenges in a standard encoder-decoder GR framework.
\textbf{Challenge 1: Semantic Similarity.}
Although shared prefixes preserve some semantic structure, randomly initialized token embeddings and a GR model trained from scratch do not directly inherit the semantic similarity relations captured by continuous item representations.
\textbf{Challenge 2: Transition Structure.}
During history encoding, user interaction histories are flattened into SID token sequences, and the encoder integrates token information through bidirectional self-attention.
However, its attention structure does not explicitly impose local constraints that respect complete SID boundaries and capture directed transitions between adjacent items.
\textbf{Challenge 3: User Preference.}
During generation, GR relies on next-token cross-entropy (CE) to maximize the likelihood of the target SID.
However, this single-target supervision does not fully exploit the set of relevant candidate items or explicitly align generation with collaborative preference relations between users and candidate items, leaving inter-item competitive relations without explicit constraints.
% GR therefore needs to explicitly model inter-item relations when learning item semantics, encoding interaction histories, and generating recommendations.
GR therefore needs to explicitly model inter-item relations when learning item semantics, encoding interaction histories, and aligning SID generation with users' relative preferences.

To jointly address these three challenges, we propose a unified \textbf{Ge}nerative \textbf{R}ecommendation framework that explicitly models semantic and collaborative \textbf{Re}lations among items, termed \textbf{REGER}.
For \textbf{Challenge 1}, we design a \emph{Generative Semantic Retrieval} (GSR) pretraining task, where the encoder takes an item's SID as input and the decoder generates the SID of a neighboring item in continuous semantic representation space.
The resulting initialization enables both SID token embeddings and the GR model to inherit inter-item semantic similarity relations before learning user preferences from interactions.
For \textbf{Challenge 2}, we propose \emph{Item Transition Structure Attention} (ITSA), which combines token-global and item-local attention heads in the history encoder.
The former capture dependencies among tokens across the entire interaction history through bidirectional self-attention, while the latter use local attention windows aligned with complete SID boundaries to jointly attend to tokens of the current item and its immediate predecessor.
This design integrates information within each item while introducing directed transition constraints, allowing global context and local transition structure to complement each other.
For \textbf{Challenge 3}, we organize SIDs into a prefix tree and establish a theoretical framework for the hierarchical decomposition of collaborative preference distributions based on the nested item sets associated with its prefixes.
We prove that the item preference distribution estimated by a discriminative recommender can be exactly factorized into conditional distributions over legal next tokens at each prefix.
We further derive a decomposition of item-level KL divergence into a weighted sum of prefix-conditional KL divergences.
Building on this theoretical analysis, we propose \emph{Hierarchical Collaborative Preference Distillation} (HCPD), which aligns the teacher's and generator's full conditional distributions over all legal next tokens at the ground-truth prefixes along the target SID path.
We jointly optimize this distillation objective with CE so that, given a user's interaction history, collaborative preference relations among candidate items directly guide SID generation.

Our main contributions are summarized as follows:
% \begin{itemize}[leftmargin=*]
%     \item We identify inherent challenges in modeling inter-item relations in GR in terms of semantic similarity, transition structure, and user preferences. We propose \textbf{REGER}, a unified generative recommendation framework, to address these challenges.
%     \item We establish a theoretical framework for hierarchical preference decomposition over SID prefix trees, proving an exact factorization of item preference distributions and a weighted decomposition of item-level KL into prefix-conditional KL divergences.
%     \item Experiments on six real-world benchmarks demonstrate consistent gains from explicit inter-item relation modeling.
% \end{itemize}

\noindent\ding{182} 
% We identify inherent challenges in modeling inter-item relations in GR in terms of semantic similarity, transition structure, and user preferences. We propose \textbf{REGER}, a unified generative recommendation framework, to address these challenges.
We identify three inherent challenges in modeling inter-item relations in GR related to semantic similarity, transition structure, and user preferences.
To address these challenges, we propose \textbf{REGER}, a unified generative recommendation framework.

\noindent\ding{183} 
% We establish a theoretical framework for hierarchical preference decomposition over SID prefix trees, proving an exact factorization of item preference distributions and a weighted decomposition of item-level KL into prefix-conditional KL divergences.
% We establish a theoretical framework that exactly factorizes item preference distributions over SID prefix trees and decomposes item-level KL into weighted prefix-conditional KL terms. Building on this, we develop the HCPD objective to guide SID generation.
We provide a theoretical analysis that yields an exact factorization of item preference distributions over SID prefix trees and a weighted decomposition of item-level KL into prefix-conditional KL terms. Building on this analysis, we develop the HCPD objective to guide SID generation.

\noindent\ding{184} 
% Experiments on six real-world benchmarks demonstrate consistent gains from explicit inter-item relation modeling and validate the effectiveness and efficiency of our proposed HCPD.
% Experiments on six real-world benchmarks demonstrate REGER’s consistent gains from explicit inter-item relation modeling and validate the effectiveness and efficiency of HCPD.
Experiments on six real-world benchmarks demonstrate REGER's consistent gains from explicit inter-item relation modeling.
Further analyses validate HCPD's effectiveness and training efficiency, while confirming that REGER retains efficient inference.
% \par

\begingroup
\brokenpenalty=100
% \par
% \vspace{-0.9 em}
\section{Related Work}
\label{Related Work}
% \noindent\textbf{Generative Recommendation}
\subsection{Generative Recommendation}
GR uses item tokenization to encode prior knowledge of inter-item relations into SIDs. TIGER~\cite{Rajput2024TIGER} pioneered the use of residual quantization in GR to learn semantic codes by reconstructing item content representations. Subsequent research has focused on incorporating collaborative signals into SID and aligning tokenization with recommendation objectives. LETTER~\cite{Wang2024LETTER} introduces collaborative alignment and code diversity regularization, while ETEGRec~\cite{Liu2025ETEGRec} alternates between tokenizer and recommender updates using alignment objectives. Other studies reconsider how item tokens are organized and represented. ActionPiece~\cite{Hou2025Actionpiece} learns merge rules over unordered feature sets using co-occurrences within items and across adjacent interactions. SST~\cite{Miao2026SST} compresses history-side item tokens into semantic subwords while retaining fixed-length SIDs as generation targets. DECOR~\cite{Liu2026DECOR} preserves token assignments, combines frozen codebook embeddings with learnable collaborative embeddings, and refines token representations through contextual composition. However, the semantic organization of SIDs does not by itself ensure that the generator learns the semantics of complete items and the relations among them. 
% REGER therefore uses explicit inter-item relations to guide both semantic pretraining and subsequent recommendation training, connecting the semantic information used for tokenization with the encoding and generation of SID sequences.

\subsection{Inter-item Relation Modeling}
% Discriminative recommenders exploit inter-item relations through transition structures, representation geometry, and ranking supervision~\cite{Wang2020GCE-GNN}. 
% % FPMC factorizes personalized transition patterns to combine sequential effects with long-term user preferences, while GCE-GNN learns item representations from both session graphs and a global transition graph. 
% Distillation further enables the transfer of relational knowledge. DCD\cite{Lee2021DCD} selects ranking knowledge based on discrepancies between teacher and student rankings, while HTD~\cite{Kang2021HTD} hierarchically transfers relational topology from the teacher's representation space.
Discriminative recommenders exploit inter-item relations through transition structures, representation geometry, and ranking supervision~\cite{Wang2020GCE-GNN,Kang2021HTD,Lee2021DCD}.
In LLM-based recommendation, explicit relation modeling is introduced through attention structures and collaborative supervision. IAM~\cite{Zhang2026IAM} alternates between attention layers restricted to tokens within each item title and across items, separating title semantics from collaborative relations. TCA~\cite{Lin2026TCA} normalizes collaborative scores over prefix-compatible items and sums the probabilities of items sharing the same next token to construct soft labels for LLM-based item-title generation. Recent work also introduces relational supervision into GR. APAO~\cite{Yu2026APAO} optimizes prefix rankings with pointwise or pairwise objectives and adaptively emphasizes underperforming prefixes, while GFlowGR~\cite{Wang2026GFlowGR} incorporates collaborative signals into item-level terminal rewards and guides SID generation through flow balance objectives.
However, GR uses discrete SID tokens shared across items as its basic modeling units, rather than continuous representations of entire items or text tokens processed by pretrained LLMs. Relations among complete items must therefore be expressed through token composition and generation, yet the connection between item-level relations and token-level dependencies remains largely implicit in GR models.
% REGER combines semantic pretraining over complete SIDs with local attention constrained by SID boundaries, and derives collaborative preference supervision over legal next tokens at each prefix through hierarchical preference decomposition. This connects item-level relations with token-level dependencies and conditional preferences during generation.

\begin{figure*}
    \centering
    \includegraphics[width=\textwidth]{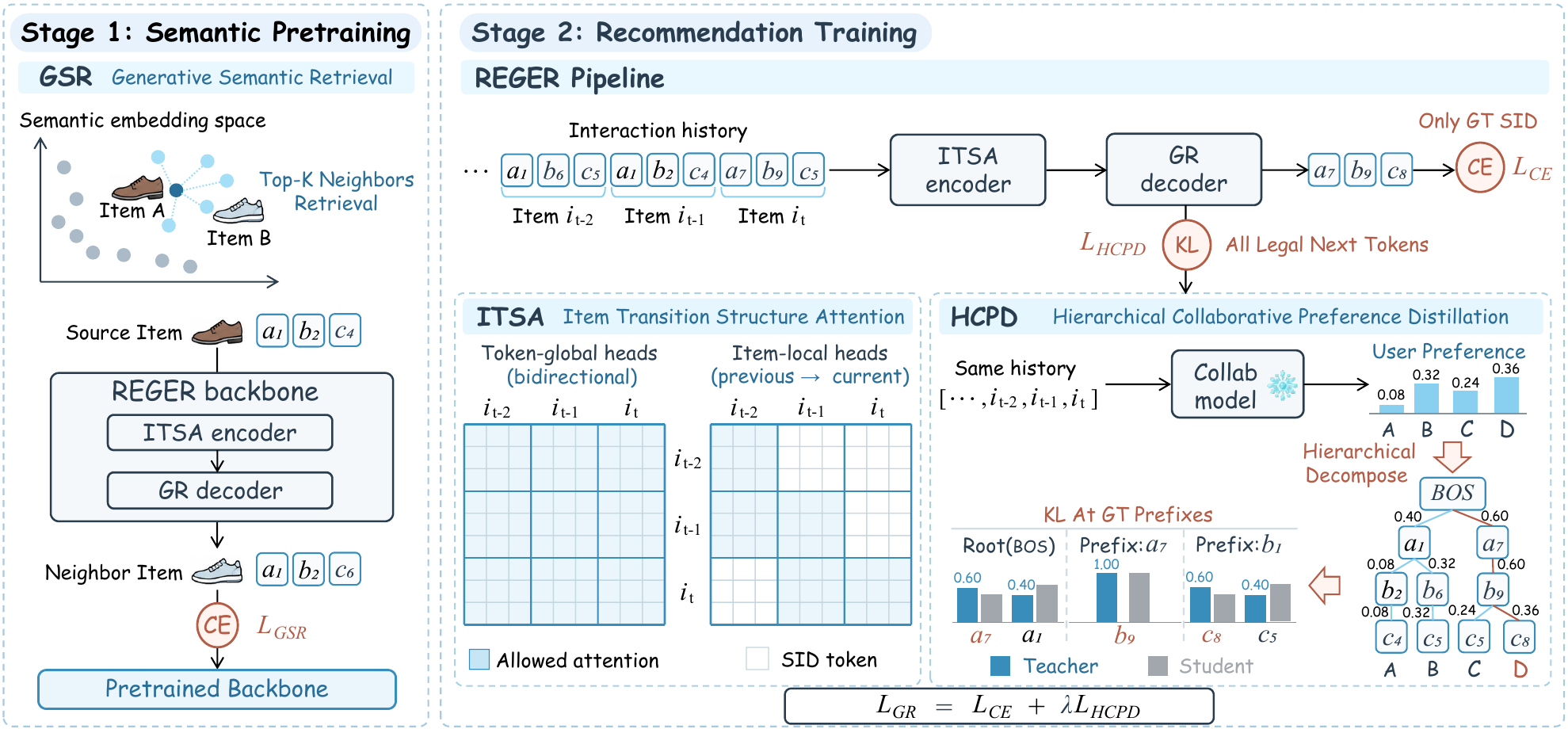}
    \caption{The overall framework of REGER.}
    % \caption{Overview of REGER. (A) GSR pretrains the backbone by generating semantic neighbors' SIDs, while (B) ITSA combines global token dependencies with local item transitions. During recommendation training, (C) HCPD hierarchically decomposes item-level collaborative preferences and distills distributions over all legal next tokens at ground-truth SID prefixes, complementing CE supervision.}
    \label{fig:framework}
\end{figure*}

\section{Preliminary: Generative Recommendation}
\label{sec:preliminaries}

\subsection{Item Tokenization}
\label{sec:item_tokenization}

Item tokenization maps continuous item representations into discrete semantic identifiers (SIDs). A widely adopted approach in academic and industrial GR applies residual quantization (RQ) to item semantic embeddings extracted using pretrained language models~\cite{Kong2025MiniOneRec,He2026PLUM}. Let $\mathbf{x}_i$ denote the semantic embedding of item $i$ in the quantization space, and let $\mathcal{C}_{\ell}=\{\mathbf{c}_{\ell,k}\}_{k=1}^{K}$ be the codebook at level $\ell$. Starting from $\mathbf{r}_0=\mathbf{x}_i$, RQ successively selects the nearest codeword and updates the residual:
\begin{equation}
s_i^{\ell}=\arg\min_{k}\|\mathbf{r}_{\ell-1}-\mathbf{c}_{\ell,k}\|_2^2,\qquad \mathbf{r}_{\ell}=\mathbf{r}_{\ell-1}-\mathbf{c}_{\ell,s_i^{\ell}}.
\label{eq:residual_quantization}
\end{equation}
After $L$ levels, the selected indices $(s_i^1,\ldots,s_i^L)$ encode item semantics in a coarse-to-fine manner.
SIDs express inter-item semantic relations through shared prefixes and codewords: shared prefixes group semantically similar items, while reused codewords serve as common semantic components in their quantized representations.
Following established practice~\cite{Rajput2024TIGER,Ju2025Handbook}, we adopt a residual-quantized variational autoencoder (RQ-VAE)~\cite{Lee2022RQVAE} for item tokenization and append an additional token to distinguish items with identical semantic codes, yielding the final SID $\mathbf{s}_i=(s_i^1,\ldots,s_i^L,s_i^{L+1})$.

\subsection{Problem Formulation}
\label{sec:problem_formulation}

Let $H_u=(i_1,\ldots,i_n)$ denote the chronologically ordered interaction history of user $u$. Each item $i$ is tokenized into a sequence of $L+1$ discrete tokens, forming its SID $\mathbf{s}_i$. Given the concatenated SIDs of $H_u$, GR predicts the next item $i_{n+1}$ by autoregressively generating its SID $\mathbf{s}_{i_{n+1}}=(s_{i_{n+1}}^1,\ldots,s_{i_{n+1}}^{L+1})$. As in large language model training, GR typically learns through next-token prediction using the \textbf{cross-entropy (CE) loss} under \textbf{teacher forcing}~\cite{Raffel2020T5,Kong2025MiniOneRec,Yu2026APAO}. Each prediction is conditioned on $H_u$ and the \textbf{ground-truth prefix} $\mathbf{s}_{i_{n+1}}^{<\ell}=(\texttt{BOS},s_{i_{n+1}}^1,\ldots,s_{i_{n+1}}^{\ell-1})$:
\begin{equation}
\begin{aligned}
\mathcal{L}_{\mathrm{GR}}=\mathcal{L}_{\mathrm{CE}}
&=-\sum_{\ell=1}^{L+1}\log p_{\theta}(s_{i_{n+1}}^{\ell}\mid H_u,\mathbf{s}_{i_{n+1}}^{<\ell})\\
&=-\sum_{\ell=1}^{L+1}\log\frac{\exp(z_{\ell,s_{i_{n+1}}^{\ell}})}{\sum_{v\in\mathcal{V}}\exp(z_{\ell,v})}.
\end{aligned}
\label{eq:gr_objective}
\end{equation}
Here, $\mathcal{V}$ denotes the SID token vocabulary, $\theta$ denotes the model parameters, and $z_{\ell,v}$ is the logit assigned to token $v$ at decoding step $\ell$.
% This objective increases the likelihood of ground-truth tokens but does not explicitly supervise fine-grained preference relations among competing tokens.
This objective uses the positive item's SID as the sole training target, without explicitly supervising fine-grained preference relations among competing tokens, leaving inter-item competitive relations to be learned implicitly.
\section{Methodology}
\label{Methodology}

% Fig.~\ref{fig:framework} presents the two-stage training pipeline of REGER. GSR first pretrains the encoder-decoder backbone to generate semantic neighbors' SIDs from source-item SIDs (Section~\ref{sec:gsr}). The encoder employs ITSA in both stages, combining token-global and item-local attention heads (Section~\ref{sec:itsa}). We derive hierarchical decompositions of item preference distributions and item-level KL divergence as the theoretical basis for HCPD (Section~\ref{sec:preference_kl_decomposition}). During recommendation training, HCPD aligns teacher and generator conditional distributions over all legal next tokens at ground-truth SID prefixes, and its distillation objective is jointly optimized with CE (Section~\ref{sec:distillation_objective}).
Fig.~\ref{fig:framework} presents the two-stage training pipeline of REGER. GSR first pretrains the encoder-decoder backbone to generate semantic neighbors' SIDs from source-item SIDs (Section~\ref{sec:gsr}). The encoder employs ITSA in both stages, combining token-global and item-local attention heads (Section~\ref{sec:itsa}). For HCPD, we first derive hierarchical decompositions of item preference distributions and item-level KL divergence. Building on these results, we align the teacher and generator conditional distributions over all legal next tokens at ground-truth SID prefixes and jointly optimize this distillation objective with CE during recommendation training (Section~\ref{sec:hcpd}).

\subsection{Generative Semantic Retrieval}
\label{sec:gsr}

% 额外补充一段分析和引入
Recent analysis shows that SIDs preserve coarse item organization but do not fully retain fine-grained neighborhood relations in the original continuous representation space~\cite{Wang2026Understanding}. 
% Although prior work reuses pretrained quantization codebook embeddings to enrich SID token representations~\cite{Liu2026DECOR}, a gap remains between quantized and original item semantic representations, and the semantic reconstruction objective used to pretrain the codebooks differs from the recommendation objective.
Although prior work reuses pretrained quantization codebook embeddings to enrich SID token representations~\cite{Liu2026DECOR}, these embeddings are pretrained for semantic reconstruction, which does not explicitly constrain fine-grained inter-item neighborhood relations 
% in the original semantic space 
or directly optimize the recommendation objective.
Moreover, advances in LLM-based GR highlight the value of aligning SID token embeddings with a pretrained LLM's semantic space before recommendation fine-tuning~\cite{TS-Rec2026Feng}. However, encoder-decoder GR models initialized from scratch lack such a pretrained semantic space.

Indeed, item content provides cues for identifying inter-item semantic relations even before user interactions are observed. Building on this intuition, we use semantic neighborhoods derived from item content to construct inter-item supervision without relying on user interaction histories, enabling the GR model to learn semantic relations among items before learning user preferences.

% To this end, we design a \emph{Generative Semantic Retrieval} (GSR) pretraining task. For each item $i$, we retrieve its top-$k$ neighbors in the continuous semantic embedding space and denote the resulting set by $\mathcal{N}_k(i)$. For each source-neighbor pair $(i,j)$ with $j \in \mathcal{N}_k(i)$, the encoder takes $\mathbf{s}_i$ as input, and the decoder autoregressively predicts the neighbor SID $\mathbf{s}_j$.
To this end, we design a \emph{Generative Semantic Retrieval} (GSR) pretraining task. 
% Let $\mathbf{e}_i$ denote the original text embedding of item $i$ produced by the pretrained language model before quantization.
Let $\mathbf{e}_i$ denote the text embedding of item $i$ used as input to the RQ-VAE tokenizer.
We retrieve its top-$k$ semantic neighbors using cosine similarity:
\begin{equation}
\mathcal{N}_k(i)
=
\operatorname*{TopK}_{j\in\mathcal{I}\setminus\{i\}}
\frac{\mathbf{e}_i^{\top}\mathbf{e}_j}
{\|\mathbf{e}_i\|_2\|\mathbf{e}_j\|_2},
\label{eq:semantic_neighbors}
\end{equation}
where $\mathcal{I}$ denotes the item catalog and $\operatorname{TopK}$ returns the item indices with the $k$ highest similarity scores. Each source-neighbor pair $(i,j)$ with $j\in\mathcal{N}_k(i)$ forms a pretraining instance: the encoder takes SID $\mathbf{s}_i$ as input, and the decoder generates the neighbor SID $\mathbf{s}_j$.

% Following the teacher-forcing formulation in Eq.~\eqref{eq:gr_objective}, we optimize GSR using the CE loss for each source-neighbor pair:
% \begin{equation}
% \mathcal{L}_{\mathrm{GSR}}(i,j)
% =
% -\sum_{\ell=1}^{L+1}
% \log p_{\theta}\!\left(
% s_j^{\ell}\mid\mathbf{s}_i,\mathbf{s}_j^{<\ell}
% \right).
% \label{eq:gsr}
% \end{equation}
Following the teacher-forcing formulation in Eq.~\eqref{eq:gr_objective}, we optimize GSR over all source-neighbor pairs using the CE loss:
\begin{equation}
\mathcal{L}_{\mathrm{GSR}}
=
-\sum_{i\in\mathcal{I}}
\sum_{j\in\mathcal{N}_k(i)}
\sum_{\ell=1}^{L+1}
\log p_{\theta}\!\left(
s_j^{\ell}\mid\mathbf{s}_i,\mathbf{s}_j^{<\ell}
\right).
\label{eq:gsr}
\end{equation}

Pretraining jointly updates the SID token embeddings and the encoder-decoder backbone, and the resulting parameters initialize subsequent recommendation training. 
To support the transfer of semantic priors to collaborative preference learning, we select the pretraining checkpoint based on recommendation performance on the validation set.
In addition, constructing multiple source-neighbor pairs with each item's complete SID as input encourages the encoder to learn intra-item token dependencies through self-attention during pretraining, while providing semantic supervision for items with sparse interactions.

\subsection{Item Transition Structure Attention}
\label{sec:itsa}

GSR provides a semantic initialization, while history encoding further requires modeling relations among items in successive interactions. In GR, interaction histories are typically flattened into SID token sequences and encoded through bidirectional self-attention, which jointly handles intra-item dependencies and inter-item behavioral transitions. SST~\cite{Miao2026SST} observes that composing item representations from their constituent tokens can consume excessive attention, resulting in \emph{intra-item attention overload}. 
% We consider it necessary to retain bidirectional interactions, since the tokens in the history must jointly represent complete items, whereas the hierarchy within SIDs does not correspond to the transition structure of user behavior. 
However, user interactions occur at the level of complete items and exhibit temporal predecessor and successor relations, calling for explicit modeling of directed inter-item transitions. 
% Meanwhile, prior work highlights recent interactions as important cues for users' current interests, alongside the complementary value of broader historical context~\cite{Kang2018SASRec}. 
Meanwhile, recent work on LLM-based recommendation highlights the complementary value of recent interactions for capturing current interests and broader historical context for modeling longer-term preferences~\cite{Qiao2025LLM4BSR,SIREN2026ICLR}.
Together, these considerations motivate explicitly emphasizing directed transitions between adjacent items while preserving intra-item semantic integration and global historical dependencies.

Inspired by these observations, we design ITSA to combine \emph{token-global} and \emph{item-local} attention heads within each encoder layer of the GR backbone, with a specified ratio between the two head types. We retain bidirectional interactions within each SID because its tokens jointly represent a complete item, and their quantization order does not encode the temporal order of user interactions. The two types of heads differ in their attention masks. Following the T5 backbone~\cite{Raffel2020T5}, the output of head $h$ is computed as
\begin{equation}
\mathbf{O}_h
=
\operatorname{softmax}\!\left(
\mathbf{Q}_h\mathbf{K}_h^{\top}
+\mathbf{B}_h+\mathbf{M}_h
\right)\mathbf{V}_h,
\label{eq:itsa_attention}
\end{equation}
where $\mathbf{Q}_h$, $\mathbf{K}_h$, and $\mathbf{V}_h$ are the query, key, and value projections of the encoder hidden states, $\mathbf{B}_h$ is the relative position bias, and $\mathbf{M}_h$ is the attention mask.

\textbf{Token-global heads.}
These heads apply standard bidirectional self-attention over the entire interaction history. For valid query and key positions $p$ and $q$, we set $M^{\mathrm{global}}_{pq}=0$, allowing each token to integrate intra-item information and capture inter-item dependencies across the full history. Direct interactions between tokens from nonadjacent items provide broader context for modeling user preferences, complementing the directed transitions captured by item-local heads.

\textbf{Item-local heads.}
These heads restrict attention to the complete SIDs of the current item and its immediate predecessor. Since the internal order of SID tokens does not encode predecessor-successor relations between user interactions, we preserve bidirectional interactions within each SID, allowing every token to incorporate information from the other tokens of its own item and thereby supporting semantic integration within both the preceding and current items. Let $g(p)$ map token position $p$ to the corresponding item position in $H_u$. The local mask is defined as
\begin{equation}
M^{\mathrm{local}}_{pq}
=
\begin{cases}
0, & g(q)\in\{g(p)-1,g(p)\},\\
-\infty, & \text{otherwise}.
\end{cases}
\label{eq:itsa_local_mask}
\end{equation}
Thus, for $t>1$, each token of item $i_t$ can attend only to tokens of $i_t$ and its immediate predecessor $i_{t-1}$, with attention weights normalized jointly over the two complete SIDs. For the first item, the window contains only its own SID. This mask provides a directed local pathway from $i_{t-1}$ to $i_t$, introducing an explicit structural bias toward short-range inter-item transitions.

The outputs of both head types are concatenated and passed through the standard output projection to integrate global historical context with local transition structure. All heads mask padding positions. Both head types have the same parameter count per head, and ITSA changes only the masks of the existing heads, leaving the GR backbone's parameter count unchanged. As shown in Fig.~\ref{fig:framework}, the same ITSA encoder is used in both stages: during GSR, both head types attend within the single source-item SID, whereas during recommendation training, item-local heads additionally attend to the immediately preceding SID.

\subsection{Hierarchical Collaborative Preference Distillation}
\label{sec:hcpd}

In this section, we first analyze hierarchical decompositions of item preference distributions and their KL divergence over the SID prefix tree, then introduce HCPD as a tractable distillation objective along ground-truth SID paths.

\subsubsection{Preference and KL Decomposition}
\label{sec:preference_kl_decomposition}

Item SIDs can be organized into a prefix tree, where each leaf corresponds to a unique item and the path from the root specifies its complete SID~\cite{Hou2026Latte}. During decoding, the root corresponds to the initial \texttt{BOS} token, and each internal node represents a SID prefix.

For a prefix $\mathbf{a}$, let $\mathcal{I}(\mathbf{a})$ denote the set of items whose SIDs extend this prefix, with $\mathcal{I}(\texttt{BOS})=\mathcal{I}$. For each nonterminal prefix, let $\mathcal{V}(\mathbf{a})$ denote the set of legal next tokens, and let $\mathbf{a}\oplus v$ denote the prefix obtained by appending token $v$. The sets $\{\mathcal{I}(\mathbf{a}\oplus v)\}_{v\in\mathcal{V}(\mathbf{a})}$ form a disjoint partition of $\mathcal{I}(\mathbf{a})$.

Given a user interaction history $H_u$, let $P$ denote a strictly positive item preference distribution over $\mathcal{I}$. We omit $H_u$ from the notation for brevity. The probability mass of a prefix and the conditional probability of a legal next token are defined as
\begin{equation}
m_P(\mathbf{a})
=
\sum_{i\in\mathcal{I}(\mathbf{a})}P(i),
\qquad
P(v\mid\mathbf{a})
=
\frac{m_P(\mathbf{a}\oplus v)}{m_P(\mathbf{a})}.
\label{eq:prefix_preference}
\end{equation}
Each conditional distribution therefore assigns probability mass to the item subsets corresponding to alternative next tokens.

\begin{theorem}[Exact preference factorization]
\label{thm:preference_factorization}
For any item $i$ with SID $\mathbf{s}_i=(s_i^1,\ldots,s_i^{L+1})$, its preference probability can be exactly factorized as
\begin{equation}
P(i)
=
\prod_{\ell=1}^{L+1}
P\!\left(s_i^\ell\mid\mathbf{s}_i^{<\ell}\right).
\label{eq:preference_factorization}
\end{equation}
\end{theorem}

This factorization expresses item preferences as conditional distributions over legal next tokens. Applying the chain rule of KL divergence to the prefix tree yields the following decomposition~\cite{Cover2026Elements}. Let $Q$ denote another strictly positive item distribution for the same history, with its prefix masses and conditional distributions defined in the same way. Let $\mathcal{P}$ denote the set of all nonterminal prefixes, including the root.

\begin{theorem}[Hierarchical KL decomposition]
\label{thm:kl_decomposition}
The KL divergence between $P$ and $Q$ can be exactly decomposed as
\begin{equation}
D_{\mathrm{KL}}(P\|Q)
=
\sum_{\mathbf{a}\in\mathcal{P}}
m_P(\mathbf{a})\,
D_{\mathrm{KL}}\!\left(
P(\cdot\mid\mathbf{a})
\|
Q(\cdot\mid\mathbf{a})
\right).
\label{eq:kl_decomposition}
\end{equation}
\end{theorem}

Proofs of both theorems are provided in Appendix~\ref{app:hcpd_theory}.

The item-level KL divergence is therefore a weighted sum of conditional KL divergences over all nonterminal prefixes, where each weight is the corresponding prefix probability mass under $P$. These results connect item-level preferences to conditional next-token distributions and motivate conditional preference matching. HCPD uses a tractable surrogate that evaluates these conditional terms along ground-truth SID paths, with its expected prefix weighting clarified below.

\subsubsection{Distillation and Training Objectives}
\label{sec:distillation_objective}

We use a frozen discriminative recommender as the teacher to provide collaborative preferences over candidate items.
For this purpose, we seek a teacher with distinguishable item representations that capture user-specific preference differences, so that its scores reflect the relative relevance of competing candidates given each user's interaction history.
Following Pctx~\cite{Zhong2025Pctx} in the choice of auxiliary recommender, we adopt DuoRec~\cite{Qiu2022DuoRec}, which uses contrastive learning to mitigate representation degeneration. We analyze the effect of teacher choice in Appendix~\ref{sec:teacher_analysis}.

Given the same interaction history $H_u$, let $r_u(i)$ denote its score for item $i$. The teacher's item preference distribution is
\begin{equation}
P_T(i)
=
\frac{\exp(r_u(i))}
{\sum_{j\in\mathcal{I}}\exp(r_u(j))}.
\label{eq:teacher_item_preference}
\end{equation}
Following Eq.~\eqref{eq:prefix_preference}, we obtain the teacher's conditional preference over each legal next token:
\begin{equation}
q_T(v\mid\mathbf{a})
=
\frac{
\sum_{i\in\mathcal{I}(\mathbf{a}\oplus v)}P_T(i)
}{
\sum_{i\in\mathcal{I}(\mathbf{a})}P_T(i)
},
\qquad v\in\mathcal{V}(\mathbf{a}).
\label{eq:teacher_token_preference}
\end{equation}
This conditional distribution transfers the teacher's relative preferences among complete items to the alternative branches available at each prefix.

For distillation, we normalize the generator's probabilities over the same legal next-token set:
\begin{equation}
q_{\theta}(v\mid\mathbf{a})
=
\frac{
p_{\theta}(v\mid H_u,\mathbf{a})
}{
\sum_{w\in\mathcal{V}(\mathbf{a})}
p_{\theta}(w\mid H_u,\mathbf{a})
},
\qquad v\in\mathcal{V}(\mathbf{a}).
\label{eq:student_token_preference}
\end{equation}
These normalized transitions induce an item distribution
$Q_{\theta}(i)=\prod_{\ell=1}^{L+1}q_{\theta}(s_i^\ell\mid\mathbf{s}_i^{<\ell})$
over the catalog, allowing Theorem~\ref{thm:kl_decomposition} to be applied to $P_T$ and $Q_{\theta}$.
This distribution is used for distillation analysis, while inference uses the full-vocabulary probabilities $p_\theta$.

\textbf{Hierarchical Collaborative Preference Distillation.}
For the target item $i_{n+1}$, let $\mathbf{a}_{\ell}=\mathbf{s}_{i_{n+1}}^{<\ell}$ denote its ground-truth prefix at decoding step $\ell$.
We align the teacher and generator distributions at these prefixes using
\begin{equation}
\mathcal{L}_{\mathrm{HCPD}}
=
\frac{1}{L+1}
\sum_{\ell=1}^{L+1}
D_{\mathrm{KL}}\!\left(
q_T(\cdot\mid\mathbf{a}_{\ell})
\|
q_{\theta}(\cdot\mid\mathbf{a}_{\ell})
\right).
\label{eq:hcpd_objective}
\end{equation}
% Each conditional KL compares the full distributions over all legal next tokens, including alternatives to the ground-truth token. Legal branches represent disjoint subsets of catalog items, so matching their probabilities provides explicit supervision on inter-item collaborative preferences. The target SID selects which prefixes receive supervision, while the teacher specifies preferences among their competing branches. For a fixed history, the expected loss weights each conditional KL by its prefix probability under the observed-target distribution, scaled by $1/(L+1)$. These probabilities generally differ from the teacher prefix masses in Eq.~\eqref{eq:kl_decomposition} with $P=P_T$. Thus, the ground-truth-path objective does not generally reproduce the scaled full-tree KL, as detailed in Appendix~\ref{app:full_tree_path_distillation}.
Each conditional KL compares the full distributions over all legal next tokens, including alternatives to the ground-truth token. Legal branches represent disjoint subsets of catalog items, so matching their probabilities provides explicit supervision on inter-item collaborative preferences. The target SID selects which prefixes receive supervision, while the teacher specifies preferences among their competing branches. Ground-truth paths align distillation with CE, concentrating preference supervision on the prefixes of observed next-item targets. Given a fixed interaction history, the expected loss weights each conditional KL by its prefix probability under the observed-target distribution, scaled by $1/(L+1)$. These probabilities generally differ from the teacher prefix masses in Eq.~\eqref{eq:kl_decomposition} with $P=P_T$. Thus, the ground-truth-path objective does not generally reproduce the scaled full-tree KL, as detailed in Appendix~\ref{app:full_tree_path_distillation}.

\textbf{Joint training objective.}
We jointly optimize the HCPD objective and the CE objective in Eq.~\eqref{eq:gr_objective}:
\begin{equation}
\mathcal{L}_{\mathrm{GR}}
=
\mathcal{L}_{\mathrm{CE}}
+
\lambda\mathcal{L}_{\mathrm{HCPD}},
\label{eq:joint_recommendation_objective}
\end{equation}
where $\lambda$ controls distillation strength.
% Notably, CE serves as the primary recommendation objective, with HCPD providing auxiliary preference supervision during training.
At inference, we use unconstrained beam search with cumulative token log-probabilities from the full-vocabulary softmax, without the teacher.
\begin{table*}[t]
\centering
\caption{
Performance comparison on the Amazon Reviews 2014 benchmarks.
Results are reported as mean $\pm$ (100 $\times$ sample standard deviation)
over three random seeds.
Bold and underlined values indicate the best and second-best results, respectively.
DuoRec serves as the frozen teacher for HCPD.
Results on the Amazon Reviews 2023 benchmarks are reported in Appendix~\ref{app:additional_results}.
}
\label{tab:main_results}

\resizebox{\linewidth}{!}{%
\begin{tabular}{l*{12}{c}}
\toprule
\multirow{2}{*}{\textbf{Methods}}
& \multicolumn{4}{c}{\textbf{Beauty}}
& \multicolumn{4}{c}{\textbf{Sports}}
& \multicolumn{4}{c}{\textbf{Toys}} \\
\cmidrule(lr){2-5}
\cmidrule(lr){6-9}
\cmidrule(lr){10-13}
& R@5 & R@10 & N@5 & N@10
& R@5 & R@10 & N@5 & N@10
& R@5 & R@10 & N@5 & N@10 \\

\midrule
\multicolumn{13}{c}{\emph{Discriminative Recommendation}} \\

SASRec
& \mstd{0.0411}{0.24} & \mstd{0.0641}{0.27} & \mstd{0.0260}{0.16} & \mstd{0.0333}{0.17}
& \mstd{0.0224}{0.09} & \mstd{0.0354}{0.08} & \mstd{0.0146}{0.07} & \mstd{0.0188}{0.08}
& \mstd{0.0497}{0.11} & \mstd{0.0727}{0.21} & \mstd{0.0333}{0.11} & \mstd{0.0408}{0.13} \\

FDSA
& \mstd{0.0240}{0.17} & \mstd{0.0397}{0.26} & \mstd{0.0152}{0.08} & \mstd{0.0202}{0.11}
& \mstd{0.0153}{0.18} & \mstd{0.0255}{0.27} & \mstd{0.0101}{0.13} & \mstd{0.0133}{0.15}
& \mstd{0.0321}{0.05} & \mstd{0.0479}{0.18} & \mstd{0.0217}{0.09} & \mstd{0.0268}{0.13} \\

BERT4Rec
& \mstd{0.0272}{0.06} & \mstd{0.0429}{0.03} & \mstd{0.0176}{0.06} & \mstd{0.0226}{0.05}
& \mstd{0.0118}{0.01} & \mstd{0.0202}{0.07} & \mstd{0.0073}{0.01} & \mstd{0.0101}{0.03}
& \mstd{0.0224}{0.13} & \mstd{0.0360}{0.19} & \mstd{0.0147}{0.09} & \mstd{0.0191}{0.11} \\

S$^3$-Rec
& \mstd{0.0378}{0.09} & \mstd{0.0634}{0.19} & \mstd{0.0236}{0.08} & \mstd{0.0319}{0.11}
& \mstd{0.0235}{0.12} & \mstd{0.0379}{0.21} & \mstd{0.0155}{0.12} & \mstd{0.0201}{0.15}
& \mstd{0.0362}{0.26} & \mstd{0.0584}{0.34} & \mstd{0.0229}{0.18} & \mstd{0.0300}{0.19} \\

\rowcolor{teacherblue}
DuoRec
& \mstd{\underline{0.0580}}{0.03} & \mstd{\textbf{0.0875}}{0.11} & \mstd{\underline{0.0358}}{0.04} & \mstd{\underline{0.0454}}{0.00}
& \mstd{\underline{0.0326}}{0.12} & \mstd{\underline{0.0493}}{0.19} & \mstd{0.0194}{0.02} & \mstd{0.0248}{0.03}
& \mstd{\textbf{0.0664}}{0.09} & \mstd{\textbf{0.0961}}{0.07} & \mstd{\textbf{0.0397}}{0.08} & \mstd{\underline{0.0493}}{0.04} \\

\midrule
\multicolumn{13}{c}{\emph{Generative Recommendation}} \\

TIGER
& \mstd{0.0435}{0.09} & \mstd{0.0673}{0.16} & \mstd{0.0291}{0.09} & \mstd{0.0367}{0.11}
& \mstd{0.0283}{0.17} & \mstd{0.0432}{0.26} & \mstd{0.0183}{0.12} & \mstd{0.0231}{0.15}
& \mstd{0.0422}{0.45} & \mstd{0.0659}{0.57} & \mstd{0.0279}{0.32} & \mstd{0.0355}{0.35} \\

LETTER
& \mstd{0.0476}{0.13} & \mstd{0.0723}{0.11} & \mstd{0.0329}{0.11} & \mstd{0.0409}{0.08}
& \mstd{0.0264}{0.09} & \mstd{0.0409}{0.15} & \mstd{0.0170}{0.08} & \mstd{0.0216}{0.10}
& \mstd{0.0424}{0.26} & \mstd{0.0641}{0.24} & \mstd{0.0279}{0.16} & \mstd{0.0348}{0.16} \\

EAGER
& \mstd{0.0448}{0.13} & \mstd{0.0665}{0.06} & \mstd{0.0304}{0.12} & \mstd{0.0373}{0.10}
& \mstd{0.0240}{0.02} & \mstd{0.0377}{0.02} & \mstd{0.0154}{0.02} & \mstd{0.0198}{0.02}
& \mstd{0.0462}{0.14} & \mstd{0.0656}{0.34} & \mstd{0.0323}{0.02} & \mstd{0.0385}{0.08} \\

UNGER
& \mstd{0.0433}{0.11} & \mstd{0.0677}{0.01} & \mstd{0.0283}{0.06} & \mstd{0.0362}{0.03}
& \mstd{0.0259}{0.04} & \mstd{0.0408}{0.11} & \mstd{0.0166}{0.04} & \mstd{0.0214}{0.06}
& \mstd{0.0416}{0.09} & \mstd{0.0665}{0.10} & \mstd{0.0268}{0.06} & \mstd{0.0348}{0.04} \\

ReSID
& \mstd{0.0475}{0.08} & \mstd{0.0739}{0.08} & \mstd{0.0321}{0.02} & \mstd{0.0405}{0.03}
& \mstd{0.0273}{0.06} & \mstd{0.0432}{0.03} & \mstd{0.0178}{0.07} & \mstd{0.0229}{0.06}
& \mstd{0.0407}{0.14} & \mstd{0.0635}{0.26} & \mstd{0.0267}{0.12} & \mstd{0.0341}{0.16} \\

ETEGRec
& \mstd{0.0497}{0.14} & \mstd{0.0759}{0.23} & \mstd{0.0334}{0.07} & \mstd{0.0419}{0.10}
& \mstd{0.0287}{0.07} & \mstd{0.0455}{0.05} & \mstd{0.0183}{0.03} & \mstd{0.0237}{0.02}
& \mstd{0.0438}{0.10} & \mstd{0.0682}{0.13} & \mstd{0.0283}{0.09} & \mstd{0.0362}{0.08} \\

DIGER
& \mstd{0.0460}{0.13} & \mstd{0.0694}{0.10} & \mstd{0.0302}{0.11} & \mstd{0.0377}{0.10}
& \mstd{0.0263}{0.15} & \mstd{0.0409}{0.11} & \mstd{0.0169}{0.10} & \mstd{0.0216}{0.09}
& \mstd{0.0418}{0.18} & \mstd{0.0660}{0.13} & \mstd{0.0271}{0.12} & \mstd{0.0349}{0.10} \\

ActionPiece
& \mstd{0.0507}{0.07} & \mstd{0.0772}{0.03} & \mstd{0.0339}{0.04} & \mstd{0.0424}{0.02}
& \mstd{0.0304}{0.02} & \mstd{0.0473}{0.05} & \mstd{0.0200}{0.01} & \mstd{\underline{0.0254}}{0.02}
& \mstd{0.0462}{0.12} & \mstd{0.0729}{0.11} & \mstd{0.0303}{0.06} & \mstd{0.0389}{0.06} \\

TIGER-SST
& \mstd{0.0487}{0.11} & \mstd{0.0722}{0.13} & \mstd{0.0336}{0.05} & \mstd{0.0411}{0.05}
& \mstd{0.0287}{0.12} & \mstd{0.0445}{0.17} & \mstd{0.0191}{0.07} & \mstd{0.0242}{0.08}
& \mstd{0.0499}{0.10} & \mstd{0.0742}{0.02} & \mstd{0.0344}{0.09} & \mstd{0.0422}{0.06} \\

\midrule
\multicolumn{13}{c}{\emph{Hybrid Recommendation}} \\

LIGER
& \mstd{0.0532}{0.06} & \mstd{0.0830}{0.10} & \mstd{0.0348}{0.02} & \mstd{0.0444}{0.03}
& \mstd{0.0304}{0.04} & \mstd{0.0480}{0.04} & \mstd{0.0197}{0.02} & \mstd{0.0253}{0.02}
& \mstd{0.0510}{0.08} & \mstd{0.0838}{0.15} & \mstd{0.0314}{0.07} & \mstd{0.0420}{0.11} \\

MemGen-GR
& \mstd{0.0502}{0.20} & \mstd{0.0748}{0.20} & \mstd{0.0342}{0.11} & \mstd{0.0421}{0.11}
& \mstd{0.0300}{0.17} & \mstd{0.0462}{0.24} & \mstd{\underline{0.0201}}{0.12} & \mstd{0.0253}{0.14}
& \mstd{0.0546}{0.34} & \mstd{0.0796}{0.33} & \mstd{0.0383}{0.25} & \mstd{0.0464}{0.25} \\

\midrule
\multicolumn{13}{c}{\emph{Ours (GR)}} \\

\rowcolor{regerblue}
\textbf{REGER}
& \mstd{\textbf{0.0594}}{0.05} & \mstd{\underline{0.0870}}{0.19} & \mstd{\textbf{0.0395}}{0.08} & \mstd{\textbf{0.0484}}{0.13}
& \mstd{\textbf{0.0357}}{0.06} & \mstd{\textbf{0.0545}}{0.12} & \mstd{\textbf{0.0235}}{0.07} & \mstd{\textbf{0.0296}}{0.08}
& \mstd{\underline{0.0612}}{0.13} & \mstd{\underline{0.0934}}{0.07} & \mstd{\underline{0.0392}}{0.08} & \mstd{\textbf{0.0496}}{0.07} \\

\bottomrule
\end{tabular}%
}
\end{table*}

\section{Experiments}

\subsection{Experimental Setting}
\textbf{Datasets.} We evaluate REGER on Beauty, Sports, and Toys from Amazon Reviews 2014~\cite{2014Amazon}. Following standard GR protocols~\cite{Rajput2024TIGER,Hou2025Actionpiece}, we apply 5-core filtering so that each remaining user and item has at least five interactions. For each user, interactions are ordered chronologically, and input histories are truncated to the 20 most recent interactions for all methods. Dataset statistics and more details are provided in Appendix~\ref{app:datasets}. We further evaluate REGER on Instruments, Scientific, and Games from Amazon Reviews 2023~\cite{2023Amazon}, with results reported in Appendix~\ref{app:additional_results}.

\noindent\textbf{Evaluation settings.} We evaluate all methods using Recall@$K$ (R@$K$) and NDCG@$K$ (N@$K$) with $K \in \{5,10\}$. 
To ensure robust results, each test is repeated three times with different random seeds $\{2024,2025,2026\}$.
% To ensure robust results, we report the mean and
% sample standard deviation of test performance across three
% training runs with random seeds {2024, 2025, 2026}.

\noindent\textbf{Baselines.} On the Amazon Reviews 2014 benchmarks, we compare REGER with three groups of baselines:
(1) \emph{Discriminative recommendation:} SASRec~\cite{Kang2018SASRec}, FDSA~\cite{Zhang2019FDSA}, BERT4Rec~\cite{Sun2019BERT4Rec}, S$^3$-Rec~\cite{Zhou2020S3Rec} and DuoRec~\cite{Qiu2022DuoRec}; (2) \emph{Generative recommendation:} TIGER~\cite{Rajput2024TIGER}, LETTER~\cite{Wang2024LETTER}, EAGER~\cite{Wang2024EAGER}, UNGER~\cite{Xiao2026UNGER}, ReSID~\cite{Liang2026ReSID}, ETEGRec~\cite{Liu2025ETEGRec}, DIGER~\cite{Fu2026DIGER}, ActionPiece~\cite{Hou2025Actionpiece}, and TIGER-SST~\cite{Miao2026SST}; (3) \emph{Hybrid recommendation:} LIGER~\cite{Yang2025LIGER} and MemGen-GR~\cite{Ding2026MemGen-GR}. 
To evaluate the effectiveness of HCPD, we further compare it with alternative training objectives for GR: TCA~\cite{Lin2026TCA}, TCA-Aux~\cite{Lin2026TCA}, APAO-pointwise~\cite{Yu2026APAO}, GFlowGR-DB~\cite{Wang2026GFlowGR}, and GFlowGR-TB~\cite{Wang2026GFlowGR}. Descriptions of all baselines, including the additional methods evaluated on the Amazon Reviews 2023 benchmarks, are provided in Appendix~\ref{app:baselines}.

\begin{table*}[!t]
\centering
\caption{Comparison of training methods applied to TIGER on the Amazon Reviews 2014 benchmarks. Methods using a collaborative model adopt the same DuoRec. Bold and underlined values indicate the best and second-best results, respectively.}
\label{tab:training_objectives}
\resizebox{\linewidth}{!}{%
\begin{tabular}{l*{12}{c}}
\toprule
\multirow{2}{*}{\textbf{Method}}
& \multicolumn{4}{c}{\textbf{Beauty}}
& \multicolumn{4}{c}{\textbf{Sports}}
& \multicolumn{4}{c}{\textbf{Toys}} \\
\cmidrule(lr){2-5}
\cmidrule(lr){6-9}
\cmidrule(lr){10-13}
& R@5 & R@10 & N@5 & N@10
& R@5 & R@10 & N@5 & N@10
& R@5 & R@10 & N@5 & N@10 \\
\midrule

TIGER
& \mstd{0.0435}{0.09}
& \mstd{0.0673}{0.16}
& \mstd{0.0291}{0.09}
& \mstd{0.0367}{0.11}
& \mstd{0.0283}{0.17}
& \mstd{0.0432}{0.26}
& \mstd{0.0183}{0.12}
& \mstd{0.0231}{0.15}
& \mstd{0.0422}{0.45}
& \mstd{0.0659}{0.57}
& \mstd{0.0279}{0.32}
& \mstd{0.0355}{0.35} \\

$+$ TCA
& \mstd{0.0431}{0.13}
& \mstd{0.0633}{0.24}
& \mstd{0.0295}{0.08}
& \mstd{0.0360}{0.12}
& \mstd{0.0256}{0.17}
& \mstd{0.0386}{0.19}
& \mstd{0.0170}{0.10}
& \mstd{0.0212}{0.10}
& \mstd{0.0472}{0.22}
& \mstd{0.0703}{0.23}
& \mstd{\underline{0.0317}}{0.07}
& \mstd{0.0391}{0.07} \\

$+$ TCA-Aux
& \mstd{\underline{0.0480}}{0.29}
& \mstd{\underline{0.0721}}{0.35}
& \mstd{\underline{0.0325}}{0.20}
& \mstd{\underline{0.0403}}{0.21}
& \mstd{\underline{0.0289}}{0.07}
& \mstd{\underline{0.0455}}{0.05}
& \mstd{\underline{0.0188}}{0.08}
& \mstd{\underline{0.0241}}{0.05}
& \mstd{\underline{0.0475}}{0.19}
& \mstd{\underline{0.0738}}{0.28}
& \mstd{0.0315}{0.19}
& \mstd{\underline{0.0400}}{0.22} \\

$+$ APAO-pointwise
& \mstd{0.0436}{0.18}
& \mstd{0.0661}{0.21}
& \mstd{0.0294}{0.15}
& \mstd{0.0366}{0.16}
& \mstd{0.0271}{0.10}
& \mstd{0.0430}{0.08}
& \mstd{0.0177}{0.09}
& \mstd{0.0228}{0.08}
& \mstd{0.0427}{0.29}
& \mstd{0.0674}{0.49}
& \mstd{0.0284}{0.21}
& \mstd{0.0363}{0.27} \\

$+$ GFlowGR-DB
& \mstd{0.0443}{0.22}
& \mstd{0.0677}{0.32}
& \mstd{0.0297}{0.11}
& \mstd{0.0373}{0.14}
& \mstd{0.0279}{0.08}
& \mstd{0.0435}{0.09}
& \mstd{0.0183}{0.05}
& \mstd{0.0234}{0.06}
& \mstd{0.0452}{0.20}
& \mstd{0.0697}{0.23}
& \mstd{0.0296}{0.12}
& \mstd{0.0375}{0.13} \\

$+$ GFlowGR-TB
& \mstd{0.0459}{0.10}
& \mstd{0.0699}{0.17}
& \mstd{0.0306}{0.12}
& \mstd{0.0383}{0.13}
& \mstd{0.0276}{0.08}
& \mstd{0.0436}{0.16}
& \mstd{0.0178}{0.05}
& \mstd{0.0230}{0.07}
& \mstd{0.0416}{0.15}
& \mstd{0.0660}{0.11}
& \mstd{0.0271}{0.14}
& \mstd{0.0350}{0.12} \\

\midrule
\rowcolor{regerblue}
\textbf{$+$ HCPD}
& \mstd{\textbf{0.0535}}{0.16}
& \mstd{\textbf{0.0794}}{0.20}
& \mstd{\textbf{0.0361}}{0.14}
& \mstd{\textbf{0.0445}}{0.15}
& \mstd{\textbf{0.0323}}{0.05}
& \mstd{\textbf{0.0501}}{0.13}
& \mstd{\textbf{0.0212}}{0.03}
& \mstd{\textbf{0.0269}}{0.06}
& \mstd{\textbf{0.0544}}{0.24}
& \mstd{\textbf{0.0823}}{0.40}
& \mstd{\textbf{0.0350}}{0.12}
& \mstd{\textbf{0.0440}}{0.19} \\

\bottomrule
\end{tabular}%
}
\end{table*}

\vspace{-0.5em}
\subsection{Main Results}

Table~\ref{tab:main_results} presents the results on the Amazon Reviews 2014 benchmarks. We make the following key observations:

\noindent\textbf{Overall effectiveness.}
On the 2014 benchmarks, DuoRec outperforms the generative baselines in most settings, consistent with previous observations that strong discriminative recommenders can outperform generative retrieval models~\cite{Yang2025LIGER}. 
% Nevertheless, REGER outperforms all generative and hybrid baselines across the three datasets and all metrics, while also achieving better overall performance than DuoRec. 
Nevertheless, REGER outperforms all generative and hybrid baselines across the three datasets and all metrics, while also outperforming DuoRec in most settings.
Even on Toys, where its Recall remains lower than DuoRec's, REGER achieves comparable or slightly better NDCG. These results support the effectiveness of explicitly incorporating inter-item semantic and collaborative relations into GR.

\noindent\textbf{Semantic and collaborative priors.}
REGER consistently outperforms methods that incorporate semantic and collaborative information into SIDs or jointly optimize tokenizer and recommender such as UNGER and ETEGRec, across all three datasets. These results further highlight the importance of explicitly learning inter-item relations within the GR model. The value of these relations depends not only on the information preserved in SIDs but also on how the generator uses these relations during item encoding and generation.

\noindent\textbf{Learning item relations within the generator.}
ActionPiece and TIGER-SST consistently outperform TIGER, highlighting the value of modeling inter-item behavioral transitions during token-level encoding. REGER achieves further gains while using fixed-length SIDs. These results show that residual-quantized SIDs can support stronger recommendation performance through explicit inter-item relation learning within the generator.

\noindent\textbf{Comparison with hybrid recommendation.}
REGER outperforms LIGER and MemGen-GR across all three datasets and metrics while relying solely on SID generation at inference. These results highlight the value of improving the GR model's own training: semantic pretraining and recommendation training explicitly guided by inter-item semantic and collaborative relations provide an effective way to strengthen its recommendation capability.

\begin{table}[!ht]
\centering
\caption{Ablation Study. \textbf{w/o GSR} skips semantic pretraining. \textbf{w/o ITSA} replaces ITSA with bidirectional self-attention using the same number of heads. \textbf{w/o HCPD} removes preference distillation during recommendation training.}
\label{tab:ablation}
\resizebox{\linewidth}{!}{%
\begin{tabular}{l*{6}{c}}
\toprule
\multirow{2}{*}{\textbf{Variants}}
& \multicolumn{2}{c}{\textbf{Beauty}}
& \multicolumn{2}{c}{\textbf{Sports}}
& \multicolumn{2}{c}{\textbf{Toys}} \\
\cmidrule(lr){2-3}
\cmidrule(lr){4-5}
\cmidrule(lr){6-7}
& R@5 & N@5
& R@5 & N@5
& R@5 & N@5 \\
\midrule

w/o GSR
& \mstd{0.0537}{0.17}
& \mstd{0.0359}{0.09}
& \mstd{\underline{0.0346}}{0.09}
& \mstd{\underline{0.0228}}{0.07}
& \mstd{0.0557}{0.07}
& \mstd{0.0357}{0.05} \\

w/o ITSA
& \mstd{\underline{0.0566}}{0.07}
& \mstd{\underline{0.0375}}{0.04}
& \mstd{0.0340}{0.07}
& \mstd{0.0223}{0.04}
& \mstd{\underline{0.0579}}{0.13}
& \mstd{\underline{0.0369}}{0.08} \\

w/o HCPD
& \mstd{0.0514}{0.10}
& \mstd{0.0357}{0.07}
& \mstd{0.0318}{0.08}
& \mstd{0.0214}{0.03}
& \mstd{0.0502}{0.19}
& \mstd{0.0336}{0.09} \\

\midrule
\rowcolor{regerblue}
\textbf{REGER}
& \mstd{\textbf{0.0594}}{0.05}
& \mstd{\textbf{0.0395}}{0.08}
& \mstd{\textbf{0.0357}}{0.06}
& \mstd{\textbf{0.0235}}{0.07}
& \mstd{\textbf{0.0612}}{0.13}
& \mstd{\textbf{0.0392}}{0.08} \\

\bottomrule
\end{tabular}%
}
\end{table}

\subsection{Ablation Study}
\label{sec:ablation}

Table~\ref{tab:ablation} presents the results of our ablation study evaluating the contributions of GSR, ITSA, and HCPD. Removing any component consistently degrades performance across all three datasets.

% \noindent\textbf{Effect of GSR.}
% Removing GSR consistently reduces performance across all three datasets, with larger drops on Beauty and Toys. These results demonstrate the value of learning inter-item semantic relations before recommendation training, providing an effective initialization for subsequent user preference learning.

% \noindent\textbf{Effect of ITSA.}
% Replacing ITSA with standard bidirectional self-attention consistently reduces performance across all three datasets. These results highlight the value of incorporating item boundaries and directed transition structure into token-level attention for modeling inter-item behavioral relations.

% \noindent\textbf{Effect of HCPD.}
% Removing HCPD causes the largest performance drop across all three datasets, with Recall@5 decreasing by 10.9\%--18.0\%. This highlights the importance of supplementing CE's positive-item supervision with collaborative preferences and supports using item-level collaborative relations to guide SID generation.

\noindent\textbf{Effect of GSR.}
Removing GSR consistently reduces performance across all three datasets, with Recall@5 decreasing by 9.6\%, 3.1\%, and 9.0\% on Beauty, Sports, and Toys, respectively.
The smaller drop on Sports than on Toys is consistent with CoRRe's observation that Sports benefits more from collaborative co-purchase signals, whereas Toys benefits from greater emphasis on the original semantic representations~\cite{Kim2026CoRRe}.
These results show that subsequent learning of collaborative signals from user interactions still benefits from semantic initialization, demonstrating the value of learning inter-item semantic relations before recommendation training as a starting point for user preference learning.

\noindent\textbf{Effect of ITSA.}
Replacing ITSA with standard bidirectional self-attention reduces Recall@5 by 4.7\%--5.4\%.
With the same number of heads and parameters, these results highlight the value of explicitly organizing token-level attention around item boundaries and directed transitions between adjacent items.
Removing ITSA causes a larger performance drop than removing GSR on Sports, whereas the reverse holds on Beauty and Toys, indicating that the relative contributions of semantic initialization and behavioral transition modeling vary across datasets.
This pattern is consistent with the joint role of semantic and behavioral factors in user preferences~\cite{Wang2024EAGER,Zhou2025CTT}.
The consistent benefits of both components further highlight the complementary value of learning inter-item semantic relations from item content and modeling behavioral transitions in user interaction sequences.

\noindent\textbf{Effect of HCPD.}
Removing HCPD causes the largest performance drop across all three datasets, with Recall@5 decreasing by 10.9\%--18.0\%. This highlights the importance of supplementing CE's positive-item supervision with collaborative preferences and supports using item-level collaborative relations to guide SID generation.

\subsection{Comparison of Training Objectives}
\label{sec:objective_comparison}

We apply HCPD directly to TIGER and compare it with TCA, APAO-pointwise, GFlowGR-DB, and GFlowGR-TB in terms of recommendation performance and training efficiency.

\noindent\textbf{Effectiveness of HCPD.}
% Table~\ref{tab:training_objectives} shows that HCPD consistently improves TIGER and achieves the best results across all three datasets and metrics. Its advantage over TCA, TCA-Aux, and both GFlowGR variants, all using the same frozen DuoRec, highlights the importance of how collaborative preferences supervise the generator. Whereas TCA-Aux augments CE with a KL loss using full-vocabulary student probabilities, HCPD renormalizes the student distribution over legal branches at each ground-truth prefix and directly matches the teacher's conditional preferences, restricting distillation to coarse-to-fine preference relations among catalog items. We analyze their objectives and gradients in Appendix~\ref{app:tca_aux_comparison}. Differentiating Eq.~\eqref{eq:hcpd_objective} with respect to each legal token's logit yields a gradient proportional to $q_{\theta}(v\mid\mathbf{a}_{\ell})-q_T(v\mid\mathbf{a}_{\ell})$, providing a direct corrective signal to non-target branches underestimated relative to the teacher. This supervision captures relative preferences among the candidate item groups represented by competing branches, complementing CE's positive-item supervision. These results show that GR can internalize inter-item preferences through training.
Table~\ref{tab:training_objectives} shows that HCPD consistently improves TIGER and achieves the best results across all three datasets and metrics. Its advantage over TCA, TCA-Aux, and both GFlowGR variants, all using the same frozen DuoRec, highlights the importance of how collaborative preferences supervise the generator. Whereas TCA-Aux augments CE with a KL loss using full-vocabulary student probabilities, HCPD renormalizes the student distribution over legal branches at each ground-truth prefix and directly matches the teacher's conditional preferences, restricting distillation to coarse-to-fine preference relations among catalog items. We analyze their objectives and gradients below. Differentiating Eq.~\eqref{eq:hcpd_objective} with respect to each legal token's logit yields a gradient proportional to $q_{\theta}(v\mid\mathbf{a}_{\ell})-q_T(v\mid\mathbf{a}_{\ell})$, providing a direct corrective signal to non-target branches underestimated relative to the teacher. This supervision captures relative preferences among the candidate item groups represented by competing branches, complementing CE's positive-item supervision. These results show that GR can internalize inter-item preferences through training.

\begin{table}[t]
\centering
\caption{Training and inference efficiency on Beauty with both training and test batch sizes set to 256. Best epoch is the epoch of peak validation performance, averaged across three seeds. Test denotes the total inference time on the test set.}
\label{tab:training_efficiency}
\resizebox{\linewidth}{!}{%
\begin{tabular}{lcccc}
\toprule
\multicolumn{1}{c}{\multirow{2}{*}{\textbf{Method}}}
& \textbf{Train memory}
& \textbf{Time/epoch}
& \textbf{Best}
& \textbf{Test} \\
& \textbf{(MB)}
& \textbf{(s)}
& \textbf{epoch}
& \textbf{(s)} \\
\midrule

TIGER
& 3072 & 19 & 104 & 16 \\

$+$ TCA
& 3104 & 22 & 136 & 106 \\

$+$ TCA-Aux
& 3104 & 22 & 117 & 16 \\

$+$ APAO-pointwise
& 3074 & 21 & 106 & 16 \\

$+$ GFlowGR-DB
& 3664 & 228 & 116 & 16 \\

$+$ GFlowGR-TB
& 3664 & 239 & 113 & 16 \\

\midrule
\rowcolor{regerblue}
\textbf{$+$ HCPD}
& 3094 & 23 & 135 & 16 \\

\bottomrule
\end{tabular}%
}
\end{table}

\noindent\textbf{Efficiency of HCPD.}
Table~\ref{tab:training_efficiency} reports training memory usage, time per epoch, average best epoch, and test time.
Experiments are conducted on Beauty using a single NVIDIA L40 GPU in FP32.
All methods use the same TIGER backbone and evaluation metrics, with KV caching disabled.
TCA uses constrained decoding following TCA's inference design~\cite{Lin2026TCA}, while the other methods including TCA-aux use unconstrained beam search.
% Compared with TIGER, HCPD adds only 22 MB of training memory and 4 s per epoch, while retaining the same inference procedure and test time of 16 s.
Together with the consistent performance gains in Table~\ref{tab:training_objectives}, these results demonstrate that HCPD provides effective collaborative preference supervision with modest training overhead and unchanged inference speed.

\subsection{In-depth Analysis}
\label{sec:indepth_analysis}

\begin{figure}[t]
    \centering
    \includegraphics[width=\linewidth]{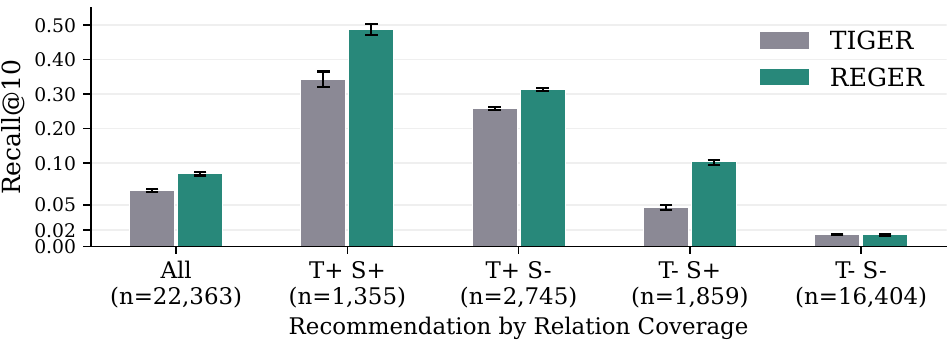}
    \caption{Semantic and transition relation coverage on Beauty. $T$ denotes whether the target is a training-set next-item successor of any history item, while $S$ denotes whether it is among the GSR Top-10 semantic neighbors of any history item. A $+$ denotes coverage; a $-$ denotes no coverage.}
    \label{fig:gsr_relationships}
\end{figure}

\noindent\textbf{Semantic and transition relation coverage.}
Fig.~\ref{fig:gsr_relationships} shows that REGER outperforms TIGER when the target is covered by both relations ($T{+}S{+}$), by a training transition alone ($T{+}S{-}$), or by a semantic relation alone ($T{-}S{+}$).
The gain in $T{-}S{+}$ is particularly pronounced: semantic proximity supports recommendation even without an observed adjacent transition.
The $T{+}S{-}$ result shows that behavioral transitions provide another useful path when the target falls outside the semantic-neighbor set.
Together, these results suggest that REGER benefits from both types of item relation.
The $T{-}S{-}$ group illustrates the limits of these two explicit relation sets.
Its targets may depend on longer-range behavior, broader semantic associations, or user-specific preferences.
We find that item relations are most useful when they provide a plausible path to the target.
Predicting a user's next interaction remains a fundamental challenge for GR when no simple inter-item relation connects the target item to the observed history.

\begin{figure}[t]
    \centering
    \includegraphics[width=\linewidth]{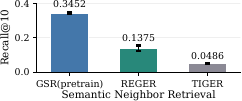}
    \caption{Top-10 semantic neighbor retrieval on Beauty.}
    \label{fig:gsr_retention}
\end{figure}

\noindent\textbf{Retention of inter-item semantic relations.}
As shown in Fig.~\ref{fig:gsr_retention}, REGER's semantic neighbor retrieval performance decreases substantially from the GSR-pretrained checkpoint.
This is consistent with the shift from retrieving static semantic neighbors to modeling user preferences conditioned on interaction histories, where semantic similarity reflects only one aspect of a user's potential interests.
Nevertheless, REGER substantially outperforms TIGER on semantic neighbor retrieval, suggesting that the generator has learned semantic relations between complete SIDs and retains part of this knowledge after subsequent collaborative preference learning.

% \begin{figure}[t]
%     \centering
%     \includegraphics[width=\linewidth]{figs/ITSA/ITSA_attention.pdf}
%     \caption{Attention patterns of ITSA on Beauty, aggregated at the item level. Position $0$ denotes the most recent item.}
%     \label{fig:itsa_attention}
% \end{figure}

% \noindent\textbf{Inter-item attention patterns.}
% Fig.~\ref{fig:itsa_attention} visualizes a token-global head and an item-local head from the final encoder layer on Beauty. We aggregate attention over each item's SID tokens and average the weights across test histories at valid position pairs. The token-global head distributes attention across the history, with relatively higher weights on recent items, providing context beyond adjacent interactions. Within the local window, the item-local head assigns greater weight to the current item while consistently attending to its immediate predecessor. The visible subdiagonal band reflects information flow from the preceding item to the current item. Together, these patterns illustrate how ITSA combines broad historical dependencies with explicit connections between adjacent items. This perspective complements SST's analysis of intra-item attention overload~\cite{Miao2026SST} by emphasizing how encoder attention can be organized to support inter-item relations.

\begin{figure}[t]
    \centering
    \includegraphics[width=\linewidth]{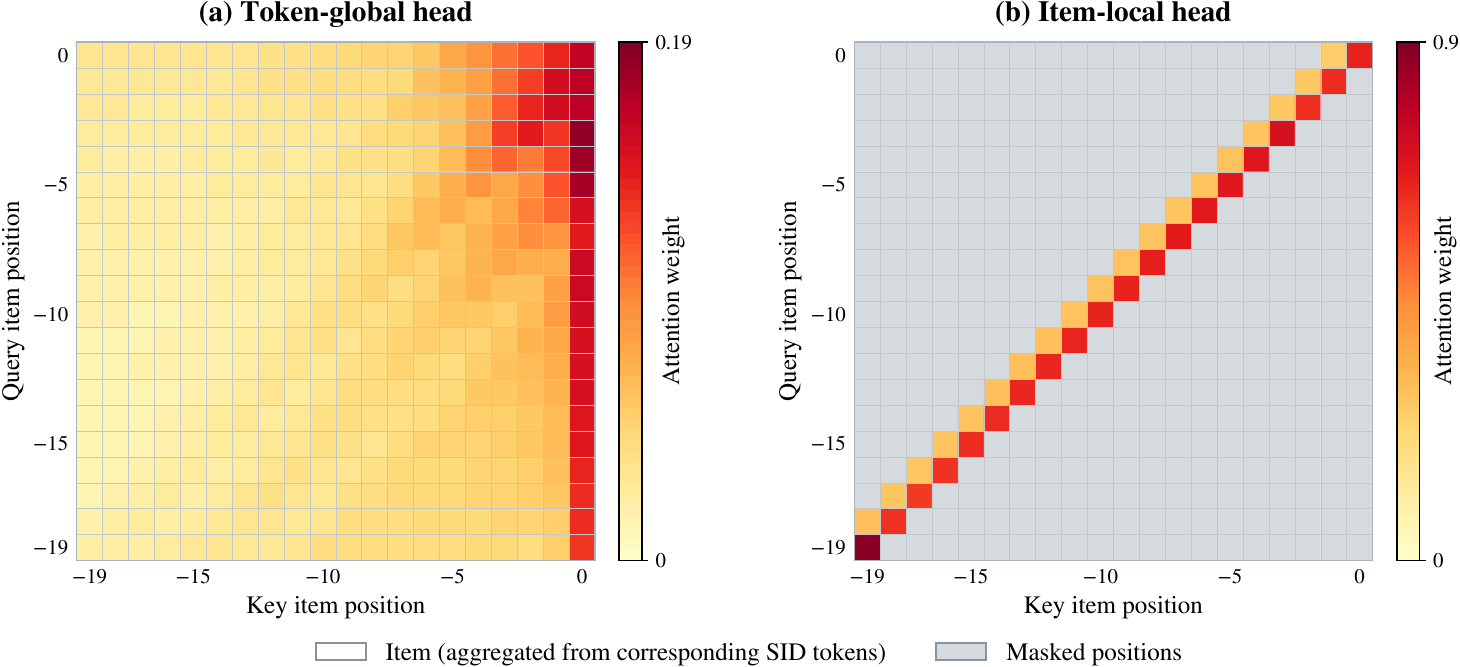}
    \caption{Attention patterns of ITSA on Beauty, aggregated at the item level. Positions $-19$ and $0$ denote the earliest and most recent items in the history window, respectively.}
    \label{fig:itsa_attention}
\end{figure}

\noindent\textbf{Inter-item attention patterns.}
Fig.~\ref{fig:itsa_attention} visualizes a token-global head and an item-local head from the final encoder layer on Beauty.
We aggregate attention over each item's SID tokens and average the weights across test histories at valid position pairs.
Across query positions, the token-global head favors recent items while maintaining attention to the rest of the history, allowing recent interactions to provide context for the encoding of earlier items.
Within the permitted local window, the item-local head assigns greater weight to the current item while consistently attending to its immediate predecessor, combining within-item information aggregation with directed information flow across SID boundaries.
Together, these patterns illustrate how ITSA encodes each item using both its immediate transition context and the broader interaction history.
This perspective complements SST's analysis of intra-item attention overload~\cite{Miao2026SST} by emphasizing how encoder attention can be organized to support inter-item relations.

\begin{figure}[t]
    \centering
    \includegraphics[width=\linewidth]{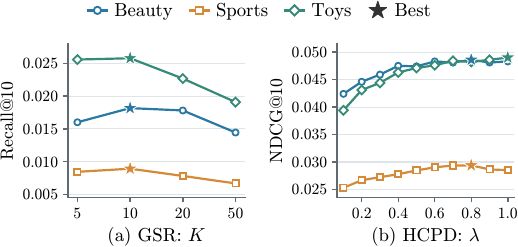}
    \caption{Effects of (a) the GSR neighborhood size $K$ on GR Recall@10 after pretraining and (b) the HCPD weight $\lambda$ on GR NDCG@10 after recommendation training.}
    \label{fig:gsr_hcpd}
\end{figure}

\begin{figure}[t]
    \centering
    \includegraphics[width=\linewidth]{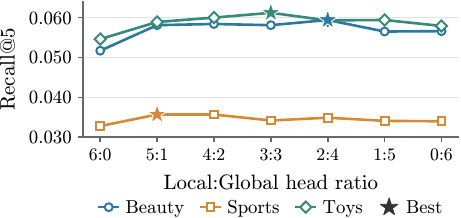}
    \caption{Effect of the item-local to token-global head ratio (Local:Global) in ITSA on Recall@5 for generative recommendation.}
    \label{fig:itsa_heads}
\end{figure}

\subsection{Hyperparameter Analysis}
\label{sec:hyperparameter_analysis}

We conduct a hyperparameter sensitivity analysis on the three Amazon Reviews 2014 datasets to examine the effects of the GSR neighborhood size $K$, the HCPD weight $\lambda$, and the item-local to token-global head ratio in ITSA.

\noindent\textbf{Effect of neighborhood size $K$.}
We vary $K\in\{5,10,20,50\}$ using four item-local heads and two token-global heads.
Fig.~\ref{fig:gsr_hcpd}(a) evaluates Stage~1 checkpoints on generative recommendation, assessing the recommendation utility of the semantic initialization.
% Recall@10 peaks at $K=10$ on all three datasets and decreases with larger neighborhoods.
Recall@10 peaks at $K=10$ on all three datasets and declines for larger neighborhoods, so we adopt $K=10$ as the default.
A small neighborhood limits the semantic relations available for pretraining, whereas an overly large neighborhood also trains the generator to predict the SIDs of weakly related items.
This may bias SID token representations and generation probabilities toward associations that do not reflect users' interaction preferences.
These results favor a focused semantic neighborhood that provides an informative starting point for subsequent collaborative preference learning, where interaction supervision refines content-derived item relations to capture user-specific preferences.
% We therefore adopt $K=10$ in our experiments.

\noindent\textbf{Effect of HCPD weight $\lambda$.}
We vary $\lambda$ from $0.1$ to $1.0$ to examine the balance between CE and collaborative preference supervision.
As shown in Fig.~\ref{fig:gsr_hcpd}(b), the best NDCG@10 within the tested range is achieved at $\lambda=0.8$ for Beauty and Sports and $\lambda=1.0$ for Toys.
Increasing $\lambda$ gives greater influence to the teacher's relative preferences among competing branches, supplementing CE's supervision on the observed target.
However, stronger distillation does not uniformly improve recommendation, as shown by the intermediate optima on Beauty and Sports.
The different preferred weights indicate that the contribution of teacher preferences should be balanced with observed-interaction supervision for each dataset.

\noindent\textbf{Effect of ITSA head ratio.}
With six attention heads per encoder layer, we vary the allocation between item-local and token-global heads while fixing $K=10$ and the selected $\lambda$ values.
Fig.~\ref{fig:itsa_heads} shows that the best tested ratios are $2:4$, $5:1$, and $3:3$ for Beauty, Sports, and Toys, respectively.
All three configurations outperform both all-local and all-global attention.
Item-local heads capture directed transitions between adjacent items, providing cues about short-term preference dynamics, while token-global heads integrate information from the full interaction history.
Their combination allows local transitions to be modeled in the context of broader user preferences.
With a fixed head count, these results highlight the value of allocating attention to complementary inter-item dependencies, with preferred ratios varying across datasets.
\section{Conclusion}
\label{sec:conclusion}

% We presented REGER, a generative recommendation framework that explicitly models inter-item semantic and collaborative relations through semantic pretraining, item transition modeling, and collaborative preference distillation. We established an exact factorization of item preference distributions and a weighted decomposition of item-level KL divergence over SID prefix trees, providing a theoretical basis for HCPD. Experiments on six real-world benchmarks demonstrate consistent performance gains, while HCPD maintains practical training efficiency. These findings highlight the value of internalizing item knowledge within the GR model itself.
We presented REGER, a generative recommendation framework that explicitly models inter-item semantic and collaborative relations through semantic pretraining, 
  item transition modeling, and collaborative preference distillation. We established an exact factorization of item preference distributions and a weighted 
  decomposition of item-level KL divergence over SID prefix trees, which motivates the design of HCPD. Experiments on six real-world benchmarks 
  demonstrate consistent performance gains, while HCPD maintains practical training efficiency. These findings highlight the value of internalizing item knowledge 
  within the GR model itself.

\bibliographystyle{ACM-Reference-Format}
\bibliography{ref}

\appendix
\appendix

\section{Experimental Settings}
\label{app:experimental_settings}

\subsection{Datasets}
\label{app:datasets}

We follow the released data preprocessing pipelines of ActionPiece\footnote{\label{fn:actionpiece}\url{https://github.com/google-deepmind/action_piece}}~\cite{Hou2025Actionpiece} and Pctx\footnote{\label{fn:pctx}\url{https://github.com/YoungZ365/Pctx}}~\cite{Zhong2025Pctx} for the 2014 and 2023 benchmarks, respectively.
For each user, the most recent interaction is used for testing, the second most recent for validation, and all earlier interactions for training.
Table~\ref{tab:dataset_statistics} summarizes the datasets after preprocessing.

% \begin{table}[H]
% \centering
% \caption{Dataset statistics after preprocessing.}
% \label{tab:dataset_statistics}
% \resizebox{\linewidth}{!}{%
% \begin{tabular}{lccccc}
% \toprule
% \textbf{Dataset}
% & \textbf{\#Users}
% & \textbf{\#Items}
% & \textbf{\#Interactions}
% & \textbf{Sparsity}
% & \textbf{Avg. Len.} \\
% \midrule
% % \multicolumn{6}{c}{\textbf{\textit{Amazon Reviews 2014}}} \\
% Beauty
% & 22,363 & 12,101 & 198,502 & 99.927\% & 8.88 \\
% Sports
% & 35,598 & 18,357 & 296,337 & 99.955\% & 8.32 \\
% Toys
% & 19,412 & 11,924 & 167,597 & 99.928\% & 8.63 \\
% \midrule
% % \multicolumn{6}{c}{\textbf{\textit{Amazon Reviews 2023}}} \\
% Instruments
% & 57,439 & 24,587 & 511,836 & 99.964\% & 8.91 \\
% Scientific
% & 50,985 & 25,848 & 412,947 & 99.969\% & 8.10 \\
% Games
% & 94,762 & 25,612 & 814,586 & 99.966\% & 8.60 \\
% \bottomrule
% \end{tabular}%
% }
% \end{table}
\begin{table}[H]
\centering
\caption{Dataset statistics after preprocessing.}
\label{tab:dataset_statistics}
\resizebox{\linewidth}{!}{%
\begin{tabular}{lccccc}
\toprule
\textbf{Dataset}
& \textbf{\#Users}
& \textbf{\#Items}
& \textbf{\#Interactions}
& \textbf{Sparsity}
& \textbf{Avg. Len.} \\
\midrule
Beauty (2014)
& 22,363 & 12,101 & 198,502 & 99.927\% & 8.88 \\
Sports (2014)
& 35,598 & 18,357 & 296,337 & 99.955\% & 8.32 \\
Toys (2014)
& 19,412 & 11,924 & 167,597 & 99.928\% & 8.63 \\
Instruments (2023)
& 57,439 & 24,587 & 511,836 & 99.964\% & 8.91 \\
Scientific (2023)
& 50,985 & 25,848 & 412,947 & 99.969\% & 8.10 \\
Games (2023)
& 94,762 & 25,612 & 814,586 & 99.966\% & 8.60 \\
\bottomrule
\end{tabular}%
}
\end{table}

\begin{table*}[t]
\centering
\caption{
Performance comparison on the Amazon Reviews 2023 benchmarks.
Published baseline results are taken from Latte~\cite{Hou2026Latte},
OMEGA~\cite{Liu2026OMEGA}, and the original papers of
DECOR~\cite{Liu2026DECOR} and TopoTok~\cite{Liu2026ToPoTok}.
The published baseline results use the same 5-core filtering and maximum history length of 20 as our experiments.
}
\label{tab:amazon2023_results}
\resizebox{\linewidth}{!}{%
\begin{tabular}{l*{12}{c}}
\toprule
\multirow{2}{*}{\textbf{Method}}
& \multicolumn{4}{c}{\textbf{Instruments}}
& \multicolumn{4}{c}{\textbf{Scientific}}
& \multicolumn{4}{c}{\textbf{Games}} \\
\cmidrule(lr){2-5}
\cmidrule(lr){6-9}
\cmidrule(lr){10-13}
& R@5 & R@10 & N@5 & N@10
& R@5 & R@10 & N@5 & N@10
& R@5 & R@10 & N@5 & N@10 \\
\midrule

\multicolumn{13}{c}{\textit{Discriminative Recommendation}} \\

SASRec
& 0.0333 & 0.0523 & 0.0213 & 0.0274
& 0.0259 & 0.0412 & 0.0150 & 0.0199
& 0.0535 & 0.0847 & 0.0331 & 0.0438 \\

BERT4Rec
& 0.0307 & 0.0485 & 0.0195 & 0.0252
& 0.0186 & 0.0296 & 0.0119 & 0.0155
& 0.0460 & 0.0735 & 0.0298 & 0.0386 \\

FDSA
& 0.0347 & 0.0545 & 0.0230 & 0.0293
& 0.0262 & 0.0421 & 0.0169 & 0.0213
& 0.0544 & 0.0852 & 0.0361 & 0.0448 \\

\rowcolor{teacherblue}
DuoRec
& \mstd{0.0377}{0.02}
& \mstd{0.0589}{0.04}
& \mstd{0.0239}{0.01}
& \mstd{0.0307}{0.01}
& \mstd{0.0275}{0.04}
& \mstd{0.0427}{0.05}
& \mstd{0.0171}{0.02}
& \mstd{0.0220}{0.01}
& \mstd{0.0582}{0.05}
& \mstd{0.0924}{0.05}
& \mstd{0.0361}{0.04}
& \mstd{0.0471}{0.03} \\

\midrule
\multicolumn{13}{c}{\textit{Generative Recommendation}} \\

TIGER
& 0.0370 & 0.0564 & 0.0244 & 0.0306
& 0.0264 & 0.0422 & 0.0175 & 0.0226
& 0.0559 & 0.0868 & 0.0366 & 0.0467 \\

LETTER
& 0.0372 & 0.0580 & 0.0246 & 0.0313
& 0.0279 & 0.0435 & 0.0182 & 0.0232
& 0.0563 & 0.0877 & 0.0372 & 0.0473 \\

ETEGRec
& 0.0395 & 0.0625 & 0.0257 & 0.0328
& 0.0288 & 0.0448 & 0.0189 & 0.0238
& 0.0589 & 0.0912 & 0.0385 & 0.0498 \\

ActionPiece
& 0.0383 & 0.0615 & 0.0243 & 0.0318
& 0.0284 & 0.0452 & 0.0182 & 0.0236
& 0.0591 & 0.0927 & 0.0382 & 0.0490 \\

DECOR
& \underline{0.0409}
& 0.0617
& \underline{0.0272}
& \underline{0.0339}
& 0.0309
& 0.0469
& \underline{0.0206}
& 0.0257
& 0.0610 & 0.0944 & 0.0400 & 0.0507 \\

Latte
& 0.0401 & 0.0618 & 0.0261 & 0.0331
& 0.0304 & 0.0470 & 0.0196 & 0.0249
& 0.0618
& 0.0958
& \textbf{0.0406}
& \textbf{0.0515} \\

TIGER-TopoTok
& 0.0402 & 0.0613 & 0.0263 & 0.0331
& 0.0302 & 0.0465 & 0.0196 & 0.0248
& 0.0608
& 0.0939
& \underline{0.0403}
& \underline{0.0509} \\

\midrule
\multicolumn{13}{c}{\textit{Hybrid Recommendation}} \\

LIGER
& 0.0408
& \underline{0.0630}
& 0.0264
& 0.0334
& \underline{0.0327}
& \textbf{0.0528}
& 0.0203
& \underline{0.0265}
& \underline{0.0619}
& \textbf{0.0969}
& 0.0391
& 0.0503 \\

\midrule
\multicolumn{13}{c}{\textit{Ours (GR)}} \\

\rowcolor{regerblue}
\textbf{REGER}
& \mstd{\textbf{0.0424}}{0.09}
& \mstd{\textbf{0.0641}}{0.10}
& \mstd{\textbf{0.0279}}{0.07}
& \mstd{\textbf{0.0349}}{0.07}
& \mstd{\textbf{0.0338}}{0.05}
& \mstd{\underline{0.0513}}{0.10}
& \mstd{\textbf{0.0221}}{0.03}
& \mstd{\textbf{0.0277}}{0.04}
& \mstd{\textbf{0.0623}}{0.01}
& \mstd{\underline{0.0966}}{0.06}
& \mstd{\textbf{0.0406}}{0.01}
& \mstd{\textbf{0.0515}}{0.03} \\

w/o GSR
& \mstd{0.0412}{0.06} & \mstd{0.0628}{0.11} & \mstd{0.0269}{0.04} & \mstd{0.0339}{0.06}
& \mstd{0.0329}{0.03} & \mstd{0.0501}{0.05} & \mstd{0.0214}{0.03} & \mstd{0.0269}{0.02}
& \mstd{0.0617}{0.03} & \mstd{0.0961}{0.03} & \mstd{0.0400}{0.05} & \mstd{0.0511}{0.05} \\

w/o ITSA
& \mstd{0.0404}{0.08} & \mstd{0.0621}{0.12} & \mstd{0.0262}{0.05} & \mstd{0.0332}{0.06}
& \mstd{0.0327}{0.06} & \mstd{0.0499}{0.00} & \mstd{0.0212}{0.04} & \mstd{0.0268}{0.02}
& \mstd{0.0617}{0.07} & \mstd{0.0961}{0.07} & \mstd{0.0402}{0.03} & \mstd{0.0512}{0.03} \\

w/o HCPD
& \mstd{0.0370}{0.05}
& \mstd{0.0561}{0.06}
& \mstd{0.0245}{0.04}
& \mstd{0.0307}{0.03}
& \mstd{0.0290}{0.04}
& \mstd{0.0448}{0.07}
& \mstd{0.0191}{0.02}
& \mstd{0.0242}{0.01}
& \mstd{0.0594}{0.09}
& \mstd{0.0910}{0.09}
& \mstd{0.0395}{0.06}
& \mstd{0.0496}{0.04} \\

\bottomrule
\end{tabular}%
}
\end{table*}

\subsection{Baselines}
\label{app:baselines}

\noindent\textbf{Main baselines.}
We organize the main baselines into three groups:
(1) \emph{Discriminative recommendation:}
SASRec~\cite{Kang2018SASRec} uses unidirectional self-attention to model interaction sequences for next-item prediction.
FDSA~\cite{Zhang2019FDSA} separately models item-level and feature-level transitions through self-attention.
BERT4Rec~\cite{Sun2019BERT4Rec} employs bidirectional self-attention with a masked item prediction objective.
S$^3$-Rec~\cite{Zhou2020S3Rec} uses self-supervised pretraining to capture correlations among attributes, items, and sequences.
DuoRec~\cite{Qiu2022DuoRec} combines dropout-based augmentation and supervised contrastive learning to alleviate representation degeneration.
(2) \emph{Generative recommendation:}
TIGER~\cite{Rajput2024TIGER} quantizes item content embeddings into SIDs using RQ-VAE and autoregressively generates the next item's SID.
LETTER~\cite{Wang2024LETTER} combines semantic reconstruction, collaborative alignment, and diversity regularization to learn item codes.
EAGER~\cite{Wang2024EAGER} uses separate behavioral and semantic generation streams with cross-stream information transfer.
UNGER~\cite{Xiao2026UNGER} integrates semantic and collaborative information into unified item codes through alignment and knowledge distillation.
ReSID~\cite{Liang2026ReSID} constructs semantic IDs from structured item features through field-aware masked autoencoding and globally aligned orthogonal quantization.
ETEGRec~\cite{Liu2025ETEGRec} alternately optimizes its tokenizer and recommender through sequence--item and preference--semantic alignment.
DIGER~\cite{Fu2026DIGER} jointly learns semantic IDs and recommendation through differentiable quantization with decaying Gumbel noise.
ActionPiece~\cite{Hou2025Actionpiece} constructs contextual tokens by merging frequently co-occurring features within items and across adjacent interactions.
TIGER-SST~\cite{Miao2026SST} applies semantic subword tokenization to TIGER's history inputs while preserving fixed-length target SIDs.
(3) \emph{Hybrid recommendation:}
LIGER~\cite{Yang2025LIGER} combines SID generation with dense retrieval to support candidate reranking and cold-start recommendation.
MemGen-GR~\cite{Ding2026MemGen-GR} adaptively combines generative and item-ID-based recommenders using a memorization-aware indicator to determine ensemble weights for each instance.

\noindent\textbf{Training objectives.}
% TCA and TCA-Aux~\cite{Lin2026TCA} maximize weighted token likelihood and augment CE with a full-vocabulary auxiliary KL loss, respectively, using token-level soft labels derived from collaborative item scores.
TCA~\cite{Lin2026TCA} derives token-level soft labels from collaborative item scores and uses them to maximize a weighted next-token likelihood.
\textbf{TCA-Aux~\cite{Lin2026TCA}  is our implementation of TCA's auxiliary KL objective on TIGER, extending the original SID experiments, which applied only the soft-NTP variant of TCA.}
APAO-pointwise~\cite{Yu2026APAO} optimizes ground-truth prefix likelihoods and adaptively emphasizes prefixes with higher losses.
GFlowGR-DB and GFlowGR-TB~\cite{Wang2026GFlowGR} impose flow-balance constraints on individual decoding transitions and complete generation trajectories, respectively, using collaborative signals to define terminal rewards.

\noindent\textbf{Additional baselines on Amazon Reviews 2023.}
DECOR~\cite{Liu2026DECOR} combines pretrained codebook embeddings with learned collaborative embeddings and adapts token representations to the interaction context.
Latte~\cite{Hou2026Latte} introduces a latent token before each SID to create multiple decoding trees and alleviate probability coupling induced by a single tree.
TopoTok~\cite{Liu2026ToPoTok} preserves semantic relationships during quantization through inter-group, intra-group, and inter-item topology distillation.

\subsection{Implementation Details}
\label{app:implementation_details}

\noindent\textbf{Item tokenization.}
Following MTGRec\footnote{\label{fn:mtgrec}\url{https://github.com/RUCAIBox/MTGRec}}~\cite{Zheng2025MTGRec} and Pctx\textsuperscript{\ref{fn:pctx}}~\cite{Zhong2025Pctx}, we encode item text using Sentence-T5~\cite{Ni2022Sentence-T5} and reduce the embeddings from 768 to 128 dimensions using PCA with whitening.
Item text comprises titles, prices, brands, and categories for Amazon Reviews 2014, and titles, features, categories, and descriptions for Amazon Reviews 2023.
We follow the default RQ-VAE configuration of MTGRec\textsuperscript{\ref{fn:mtgrec}} and train for up to 20,000 epochs.
The tokenizer is trained once per dataset with seed 2024, and its SID mapping is shared across recommender seeds.

\noindent\textbf{Backbone and baselines.}
We adopt the T5 encoder--decoder configuration used in ActionPiece\textsuperscript{\ref{fn:actionpiece}}~\cite{Hou2025Actionpiece}.
By default, reproduced GR baselines are trained for up to 200 epochs with a batch size of 256 and weight decay of 0.05.
We train each reproduced GR baseline on each dataset with learning rates of $3 \times 10^{-3}$ and $3 \times 10^{-4}$, and report the test results of the configuration achieving the highest validation NDCG@10.
We use cosine learning rate scheduling with 10,000 warm-up steps and an early stopping patience of 25 epochs, selecting checkpoints by validation NDCG@10.
The beam size is 20 for validation and testing, except for EAGER, ReSID, and ActionPiece\textsuperscript{\ref{fn:actionpiece}}, for which we use the beam sizes specified in their official implementations.
We reproduce discriminative baselines using RecBole~\cite{Zhao2021Recbole} or their official implementations.
For the remaining baselines, we follow their official implementations and reported settings.

\noindent\textbf{REGER.}
On both benchmark versions, GSR pretraining uses a learning rate of $3 \times 10^{-3}$ and runs for up to 100 epochs with an early stopping patience of 10 epochs.
Recommendation training jointly optimizes CE and HCPD for up to 200 epochs with a learning rate of $3 \times 10^{-4}$ and an early stopping patience of 15 epochs.
Both stages use AdamW with cosine learning rate scheduling and the batch size, weight decay, and checkpoint selection criterion specified above.
For all datasets, we set $K=10$ and freeze the teacher during recommendation training.
The ITSA head ratios and HCPD weights for the 2014 datasets are reported in Section~\ref{sec:hyperparameter_analysis}.
For the 2023 datasets, the item-local:token-global head ratio and $\lambda$ are $5:1$ and $0.9$ for Instruments and Scientific, and $4:2$ and $1.0$ for Games.

% \noindent\textbf{Remark.}
% Following TCA's inference settings~\cite{Lin2026TCA}, we use constrained beam search for TCA and TCA-Aux to preserve their recommendation performance, while the other GR methods and REGER use unconstrained beam search .

\noindent\textbf{Inference settings.}
Following TCA's inference settings~\cite{Lin2026TCA}, we use constrained beam search for TCA to preserve its recommendation performance. All other GR methods, including REGER and TCA-Aux, use unconstrained beam search.

\begin{table}[t]
\centering
\caption{Effect of the distillation teacher on REGER on the 2014 benchmarks. All results use a beam size of 50.}
\label{tab:teacher_comparison}
\resizebox{\linewidth}{!}{%
\begin{tabular}{l*{6}{c}}
\toprule
\multirow{2}{*}{\textbf{Method}}
& \multicolumn{2}{c}{\textbf{Beauty}}
& \multicolumn{2}{c}{\textbf{Sports}}
& \multicolumn{2}{c}{\textbf{Toys}} \\
\cmidrule(lr){2-3}
\cmidrule(lr){4-5}
\cmidrule(lr){6-7}
& N@5 & R@5
& N@5 & R@5
& N@5 & R@5 \\
\midrule

REGER w/o HCPD
& \mstd{\underline{0.0357}}{0.07}
& \mstd{\underline{0.0514}}{0.10}
& \mstd{\underline{0.0214}}{0.03}
& \mstd{\underline{0.0319}}{0.08}
& \mstd{\underline{0.0336}}{0.09}
& \mstd{\underline{0.0502}}{0.19} \\

REGER (SASRec)
& \mstd{0.0289}{0.14}
& \mstd{0.0465}{0.20}
& \mstd{0.0167}{0.03}
& \mstd{0.0285}{0.06}
& \mstd{0.0288}{0.12}
& \mstd{0.0498}{0.14} \\

\rowcolor{regerblue}
\textbf{REGER (DuoRec)}
& \mstd{\textbf{0.0395}}{0.08}
& \mstd{\textbf{0.0594}}{0.05}
& \mstd{\textbf{0.0235}}{0.07}
& \mstd{\textbf{0.0357}}{0.07}
& \mstd{\textbf{0.0392}}{0.08}
& \mstd{\textbf{0.0612}}{0.13} \\

\midrule
\multirow{2}{*}{\textbf{Method}}
& \multicolumn{2}{c}{\textbf{Beauty}}
& \multicolumn{2}{c}{\textbf{Sports}}
& \multicolumn{2}{c}{\textbf{Toys}} \\
\cmidrule(lr){2-3}
\cmidrule(lr){4-5}
\cmidrule(lr){6-7}
& N@50 & R@50
& N@50 & R@50
& N@50 & R@50 \\
\midrule

REGER w/o HCPD
& \mstd{\underline{0.0637}}{0.08}
& \mstd{0.1685}{0.22}
& \mstd{\underline{0.0410}}{0.01}
& \mstd{0.1144}{0.04}
& \mstd{\underline{0.0605}}{0.14}
& \mstd{0.1621}{0.38} \\

REGER (SASRec)
& \mstd{0.0599}{0.14}
& \mstd{\underline{0.1759}}{0.26}
& \mstd{0.0377}{0.04}
& \mstd{\underline{0.1172}}{0.09}
& \mstd{\underline{0.0605}}{0.10}
& \mstd{\underline{0.1796}}{0.20} \\

\rowcolor{regerblue}
\textbf{REGER (DuoRec)}
& \mstd{\textbf{0.0710}}{0.11}
& \mstd{\textbf{0.1911}}{0.15}
& \mstd{\textbf{0.0455}}{0.09}
& \mstd{\textbf{0.1280}}{0.15}
& \mstd{\textbf{0.0716}}{0.06}
& \mstd{\textbf{0.1943}}{0.10} \\

\bottomrule
\end{tabular}%
}
\end{table}

% \section{Results on Amazon Reviews 2023}
% \label{app:additional_results}

% REGER achieves the best or second-best results across all metrics on the larger 2023 benchmarks (Table~\ref{tab:amazon2023_results}), highlighting the value of learning inter-item collaborative preferences within GR.
% Although its teacher DuoRec trails recent generative baselines, HCPD consistently improves REGER.
% Interestingly, these gains persist on Scientific even though DuoRec underperforms REGER w/o HCPD on every metric.
% These results show that a teacher's top-$K$ recommendation performance does not fully reflect the supervisory value of its preference distribution.
% We attribute these gains to DuoRec's ability to capture each user's relative preferences among candidate items through discriminative item representations, providing supervision complementary to observed interactions. Furthermore, comparisons with TopoTok, DECOR, and Latte highlight the importance of translating inter-item relations into user-specific generation probabilities: semantically related items may differ in relevance to a user, and explicit preference supervision helps the generator distinguish among competing candidates.

\section{Results on Amazon Reviews 2023}
\label{app:additional_results}

\noindent\textbf{Performance on larger datasets.}
REGER achieves the best or second-best results across all reported metrics (Table~\ref{tab:amazon2023_results}), highlighting the value of learning inter-item collaborative preferences within GR.
Although its teacher DuoRec trails recent generative baselines, HCPD consistently improves REGER.
Interestingly, these gains persist on Scientific even though DuoRec underperforms REGER w/o HCPD on every metric.
These results show that a teacher's top-$K$ recommendation performance does not fully reflect the supervisory value of its preference distribution.
We attribute these gains to DuoRec's ability to capture each user's relative preferences among candidate items through discriminative item representations, providing supervision complementary to observed interactions.
Furthermore, comparisons with TopoTok, DECOR, and Latte highlight the importance of translating inter-item relations into user-specific generation probabilities. Semantically related items may differ in relevance to a user, and explicit preference supervision helps the generator distinguish among competing candidates.
Notably, Latte~\cite{Hou2026Latte} creates multiple generation paths for the same SID by varying the latent token prepended to it and aggregates their scores to express user-specific rankings, whereas REGER learns user-specific probabilities over legal branches of its original SID tree through collaborative preference distillation, achieving comparable or better performance without generating an additional token at inference.

\noindent\textbf{Component ablations.}
On these larger benchmarks, removing HCPD again causes the largest performance drop, with Recall@5 decreasing by 12.7\%, 14.2\%, and 4.7\% on Instruments, Scientific, and Games, respectively.
Removing either GSR or ITSA also reduces performance.
The drop from removing ITSA is larger than that from removing GSR on Instruments and Scientific, while both drops are smaller on Games.
These gains on datasets with more users and larger item catalogs support integrating relations from two heterogeneous sources, item text and interaction logs.

\section{Effect of the Distillation Teacher}
\label{sec:teacher_analysis}

We compare SASRec and DuoRec as frozen teachers against REGER w/o HCPD (Table~\ref{tab:teacher_comparison}). Although SASRec performs worse overall than the student (Table~\ref{tab:main_results}), using it as the teacher improves Recall@50 across all three datasets, demonstrating that a weaker teacher can still provide useful collaborative information. However, the resulting student exhibits lower NDCG@5 and Recall@5. KL matching (Eq.~\eqref{eq:hcpd_objective}) helps explain this pattern: it can promote relevant candidates underestimated by the student, but can also transfer inaccurate teacher preferences that weaken the student's existing ranking. DuoRec consistently improves all reported metrics, benefiting both longer-list recall and top-ranked recommendation. Together with the 2023 results, these findings highlight the importance of a teacher's ability to accurately model a user's relative preferences among candidate items. This supports our choice of DuoRec~\cite{Qiu2022DuoRec}, which uses contrastive learning to mitigate representation degeneration and learn discriminative representations.

\section{Theoretical Analysis of HCPD}
\label{app:hcpd_theory}

We use the notation and assumptions of Section~\ref{sec:hcpd}.

\subsection{Proof of Theorem~\ref{thm:preference_factorization}}
\label{app:proof_preference_factorization}

\begin{proof}
The root has probability mass $1$, and the leaf for item $i$ has mass $P(i)$.
Using Eq.~\eqref{eq:prefix_preference}, the product telescopes:
\begin{equation*}
\prod_{\ell=1}^{L+1}
P(s_i^\ell\mid\mathbf{s}_i^{<\ell})
=
\prod_{\ell=1}^{L+1}
\frac{
m_P(\mathbf{s}_i^{<\ell}\oplus s_i^\ell)
}{
m_P(\mathbf{s}_i^{<\ell})
}
=
P(i).
\qedhere
\end{equation*}
\end{proof}

\subsection{Proof of Theorem~\ref{thm:kl_decomposition}}
\label{app:proof_kl_decomposition}

\begin{proof}
By Theorem~\ref{thm:preference_factorization}, expanding the item-level KL and grouping terms by tree edges gives
\begin{equation*}
\begin{gathered}
D_{\mathrm{KL}}(P\|Q)
=
\sum_{i\in\mathcal{I}}P(i)
\sum_{\ell=1}^{L+1}
\log
\frac{
P(s_i^\ell\mid\mathbf{s}_i^{<\ell})
}{
Q(s_i^\ell\mid\mathbf{s}_i^{<\ell})
}
\\
=
\sum_{\mathbf{a}\in\mathcal{P}}
\sum_{v\in\mathcal{V}(\mathbf{a})}
m_P(\mathbf{a}\oplus v)
\log
\frac{P(v\mid\mathbf{a})}
     {Q(v\mid\mathbf{a})}
\\
=
\sum_{\mathbf{a}\in\mathcal{P}}
m_P(\mathbf{a})\,
D_{\mathrm{KL}}\!\left(
P(\cdot\mid\mathbf{a})
\|
Q(\cdot\mid\mathbf{a})
\right).
\end{gathered}
\end{equation*}
The second equality collects the items traversing each edge, whose total probability is $m_P(\mathbf{a}\oplus v)$.
The last equality uses
$m_P(\mathbf{a}\oplus v)=m_P(\mathbf{a})P(v\mid\mathbf{a})$.
\end{proof}

\subsection{Full-Tree KL and Ground-Truth Path Distillation}
\label{app:full_tree_path_distillation}

Given a fixed interaction history, let $R$ be a distribution over catalog items, with prefix masses $m_R$ defined as in Eq.~\eqref{eq:prefix_preference}.
Write $\mathcal{L}_{\mathrm{HCPD}}(i)$ for Eq.~\eqref{eq:hcpd_objective} evaluated at target item $i$, and abbreviate the conditional KL as
$d(\mathbf{a})=
D_{\mathrm{KL}}(q_T(\cdot\mid\mathbf{a})\|q_\theta(\cdot\mid\mathbf{a}))$.
Grouping target paths by prefix yields
\begin{equation}
\mathbb{E}_{i\sim R}
[\mathcal{L}_{\mathrm{HCPD}}(i)]
=
\frac{1}{L+1}
\sum_{\mathbf{a}\in\mathcal{P}}
m_R(\mathbf{a})\,d(\mathbf{a}).
\label{eq:app_path_expectation}
\end{equation}
For teacher-sampled targets ($R=P_T$), Theorem~\ref{thm:kl_decomposition} gives
\begin{equation}
\mathbb{E}_{i\sim P_T}
[\mathcal{L}_{\mathrm{HCPD}}(i)]
=
\frac{1}{L+1}
D_{\mathrm{KL}}(P_T\|Q_\theta).
\label{eq:app_teacher_sampled_path}
\end{equation}
Thus, teacher-sampled paths provide an unbiased estimator of the scaled item-level KL.
Observed targets need not follow $P_T$, so their prefix weights generally differ from those in the full-tree decomposition.
Each instance evaluates conditional KL terms only at prefixes along its ground-truth path.

\subsection{Comparison with TCA-Aux}
\label{app:tca_aux_comparison}

TCA~\cite{Lin2026TCA} motivates soft-NTP through its adaptive gradient weighting and reports that it outperforms auxiliary KL in an LLM-based item-title recommendation ablation.
\textbf{However, our GR experiments show the exact opposite trend.}
One possible explanation is that a SID generator trained from scratch must learn the relevance of discrete codes from recommendation data.
Soft-NTP weights collaborative targets by the generator's current probabilities, which can weaken supervision for continuations it underestimates, whereas auxiliary KL retains their teacher-derived target weights.
We further analyze how HCPD differs from TCA-Aux under the same teacher distribution, legal sets, and target prefixes.
% We compare HCPD and TCA-Aux~\cite{Lin2026TCA} using the same teacher distribution, legal sets, and target prefixes.
% TCA-Aux uses full-vocabulary student probabilities in its auxiliary KL.
At a fixed prefix $\mathbf{a}$, write
$t_v=q_T(v\mid\mathbf{a})$ and
$q_v=q_\theta(v\mid\mathbf{a})$
for legal tokens $v$, and suppress the conditioning on $H_u$ and $\mathbf{a}$ in $p_\theta(v)$.
Let $y\in\mathcal{V}(\mathbf{a})$ be the target token and
$Z=\sum_{v\in\mathcal{V}(\mathbf{a})}p_\theta(v)$
the student's legal probability mass.
Since $p_\theta(v)=Zq_v$ for legal tokens, extending $t$ by zero elsewhere gives
\begin{equation}
D_{\mathrm{KL}}(t\|p_\theta)
=
D_{\mathrm{KL}}(t\|q)-\log Z.
\label{eq:app_kl_mass_identity}
\end{equation}

Let $\kappa=\lambda/(L+1)>0$ account for HCPD's averaging over positions.
For TCA-Aux, let $A,B\geq0$ denote fixed weights on the full-vocabulary CE and KL terms, absorbing any positive overall rescaling.
The per-prefix joint losses of HCPD and reweighted TCA-Aux are
\begin{equation}
\begin{gathered}
\mathcal{J}_{\mathrm{HCPD}}
=-\log q_y+\kappa D_{\mathrm{KL}}(t\|q)-\log Z,
\\
\mathcal{J}_{\mathrm{aux}}
=-A\log q_y+B D_{\mathrm{KL}}(t\|q)-(A+B)\log Z.
\end{gathered}
\label{eq:app_auxiliary_objectives}
\end{equation}
Matching the coefficients of the conditional fitting terms gives
\begin{equation}
\left.\mathcal{J}_{\mathrm{aux}}\right|_{A=1,\,B=\kappa}
=
\mathcal{J}_{\mathrm{HCPD}}-\kappa\log Z.
\label{eq:app_objective_difference}
\end{equation}
The fitting terms within the legal set are then identical, but TCA-Aux imposes an additional penalty for insufficient legal probability mass.

With at least one illegal token, unrestricted logits allow $q$ and $Z$ to vary independently.
If $t$ assigns positive mass to a non-target legal token, equality up to a student-independent constant requires
$A=1$, $B=\kappa$, and $A+B=1$.
These conditions contradict $\kappa>0$, ruling out equivalence by fixed reweighting or overall rescaling in general.

The distinction also appears in the KL gradients.
Let $z_v$ denote the student's logit before the full-vocabulary softmax.
With the teacher fixed, the gradients for a legal token $v$ are
\begin{equation}
\frac{\partial D_{\mathrm{KL}}(t\|q)}{\partial z_v}
=q_v-t_v,
\qquad
\frac{\partial D_{\mathrm{KL}}(t\|p_\theta)}{\partial z_v}
=Zq_v-t_v.
\label{eq:app_distillation_gradients}
\end{equation}
For illegal tokens, these KL derivatives are $0$ and $p_\theta(v)$, respectively.
HCPD's joint objective still penalizes their probability through full-vocabulary CE.

When $q=t$ but $Z<1$, HCPD's KL and its logit gradients vanish, whereas TCA-Aux's KL remains $-\log Z$.
At prefixes with a single legal continuation, HCPD's KL is identically zero, since no competing item groups need to be distinguished.

\section{Scope of KL Supervision in HCPD}
\label{app:hcpd_design}

We evaluate \textbf{HCPD-Full}, a variant of REGER, on Beauty to examine the normalization design analyzed in Appendix~\ref{app:tca_aux_comparison}.
It retains ground-truth distillation paths but computes KL using full-vocabulary student probabilities, extending the teacher distribution by zero outside the legal set.

\begin{table}[t]
\centering
\caption{Comparison of HCPD variants for REGER on Beauty under the same hyperparameter settings.}
\label{tab:hcpd_design}
\resizebox{\linewidth}{!}{%
\begin{tabular}{lcccccc}
\toprule
\textbf{Method}
& \textbf{N@5}
& \textbf{R@5}
& \textbf{N@10}
& \textbf{R@10}
& \textbf{N@20}
& \textbf{R@20} \\
\midrule

HCPD-Full
& \mstd{0.0382}{0.13}
& \mstd{0.0580}{0.09}
& \mstd{0.0475}{0.14}
& \mstd{0.0867}{0.15}
& \mstd{0.0574}{0.16}
& \mstd{0.1241}{0.27} \\

\rowcolor{regerblue}
\textbf{HCPD}
& \mstd{0.0395}{0.08}
& \mstd{0.0594}{0.05}
& \mstd{0.0484}{0.13}
& \mstd{0.0870}{0.19}
& \mstd{0.0578}{0.12}
& \mstd{0.1243}{0.18} \\

\bottomrule
\end{tabular}%
}
\end{table}

Table~\ref{tab:hcpd_design} shows that HCPD achieves higher mean performance than HCPD-Full across all six metrics, with the largest relative improvements in NDCG@5 and Recall@5, approximately 3.4\% and 2.4\%, respectively.
As the recommendation list length increases, the performance gaps narrow, with nearly identical Recall@20, suggesting that HCPD better prioritizes relevant items near the top of the list.
HCPD-Full's competitive performance for longer recommendation lists is consistent with the conditional preference-fitting term shared by both objectives.
By Eq.~\eqref{eq:app_kl_mass_identity}, its full-vocabulary KL additionally penalizes insufficient legal probability mass, increasing the weight of a penalty already present in CE.
Using the notation of that appendix, this additional penalty has gradient $\partial(-\log Z)/\partial z_v=-(1-Z)q_v$ for a legal token $v$.
This teacher-independent term favors larger logit increases for branches already receiving higher student probabilities and can counteract corrections to overestimated branches.
Together with the comparison against TCA-Aux in Table~\ref{tab:training_objectives}, these results highlight the value of HCPD's conditional KL as a complement to CE: CE anchors learning to the observed target, while normalization over legal continuations focuses distillation on relative preferences among competing catalog items.

\end{document}